\documentclass[onefignum,onetabnum]{siamart171218}

\usepackage{lipsum}
\usepackage{amsfonts}
\usepackage{graphicx,xspace}
\usepackage{epstopdf}
\usepackage{algorithmic}
\ifpdf
  \DeclareGraphicsExtensions{.eps,.pdf,.png,.jpg}
\else
  \DeclareGraphicsExtensions{.eps}
\fi

\newsiamremark{remark}{Remark}
\newsiamremark{hypothesis}{Hypothesis}
\crefname{hypothesis}{Hypothesis}{Hypotheses}
\newsiamthm{claim}{Claim}

\headers{Finite-Time Influence Systems \\ and the Wisdom Of Crowd Effect}{Francesco Bullo, Fabio Fagnani and Barbara Franci}

\title{Finite-Time Influence Systems \\ and the Wisdom Of Crowd Effect\thanks{Submitted to the editors DATE.
\funding{This material is
  based upon work supported by, or in part by, the U.S.\ Army
  Research Laboratory and the U.S.\ Army Research Office under grant
  number W911NF-15-1-0577. This work has been done while Barbara Franci was a PhD Student sponsored by a TIM/Telecom Italia grant.
}}}

\author{Francesco Bullo\thanks{Mechanical
  Engineering Department and the Center of Control, Dynamical-Systems
  and Computation, UC Santa Barbara, CA 93106-5070,
  USA. \email{bullo@engineering.ucsb.edu}}
\and Fabio Fagnani\thanks{Department
of Mathematical Sciences, Politecnico di Torino, Torino, Italy. \email{fabio.fagnani@polito.it}}
\and Barbara Franci\footnotemark[3]\thanks{Department
of Mathematical Sciences, Politecnico di Torino, Torino, Italy. \email{barbara.franci@polito.it.}}}

\usepackage{amsopn}
\DeclareMathOperator{\diag}{diag}

\ifpdf
\hypersetup{
  pdftitle={Finite-Time Influence Systems \\ and the Wisdom Of Crowd Effect},
  pdfauthor={Francesco Bullo, Fabio Fagnani and Barbara Franci}
}
\fi

\graphicspath{{./fig/}}

\newtheorem{example}{Example}

\newcommand{\real}{\mathbb{R}}
\newcommand{\realnonnegative}{\mathbb{R}_{\geq0}}

\newcommand{\integernonnegative}{\mathbb{Z}_{\geq0}}
\def\N{\mathbb{N}}
\def\natural{\mathbb{N}}
\newcommand{\mc}{\mathcal}

\newcommand{\onenorm}[1]{\|#1\|_1}
\newcommand{\bigonenorm}[1]{\big\|#1\big\|_1}
\newcommand{\Bigonenorm}[1]{\Big\|#1\Big\|_1}
\newcommand{\inftynorm}[1]{\|#1\|_{\infty}}

\newcommand{\simplex}[1]{\Delta_{#1}}
\newcommand{\taumix}{\subscr{\tau}{mix}}

\DeclareMathOperator{\EE}{{\mathbb{E}}}
\DeclareMathOperator{\Var}{Var}
\DeclareMathOperator{\Prob}{\mathbb{P}}

\newcommand{\e}{\textrm{e}}

\newcommand{\enne}[1]{#1^{[n]}}
\DeclareMathOperator{\ave}{ave}

\newcommand{\BB}{\mc B}

\newcommand{\until}[1]{\{1,\dots, #1\}}
\newcommand{\fromto}[2]{\{#1,\dots, #2\}}
\newcommand{\subscr}[2]{#1_{\textup{#2}}}
  \newcommand{\supscr}[2]{#1^{[#2]}}
\newcommand{\setdef}[2]{\{#1 \; | \; #2\}}

\newcommand{\bigsetdef}[2]{\big\{#1 \; \big| \; #2\big\}}
\newcommand{\Bigsetdef}[2]{\Big\{#1 \; \Big| \; #2\Big\}}
\newcommand{\map}[3]{#1: #2 \rightarrow #3}

\newcommand\oprocendsymbol{\hbox{$\square$}}
\newcommand\oprocend{\relax\ifmmode\else\unskip\hfill\fi\oprocendsymbol}

\renewcommand{\theenumi}{(\roman{enumi})}

\DeclareSymbolFont{bbold}{U}{bbold}{m}{n}
\DeclareSymbolFontAlphabet{\mathbbold}{bbold}
\newcommand{\vect}[1]{\mathbbold{#1}}

\newcommand{\ds}{\displaystyle}

\newcounter{saveenum}

\newcommand{\Barabasi}{Barab{\'a}si}
\newcommand{\Erdos}{Erd{\"o}s}
\newcommand{\Renyi}{R{\'e}nyi}
\newcommand{\naive}{na\"ive\xspace}

\begin{document}

\maketitle

\begin{abstract}
  Recent contributions have studied how an influence system may affect the
  wisdom of crowd phenomenon. In the so-called na\"ive learning setting, a
  crowd of individuals holds opinions that are statistically independent
  estimates of an unknown parameter; the crowd is wise when the average
  opinion converges to the true parameter in the limit of infinitely many
  individuals. Unfortunately, even starting from wise initial opinions, a
  crowd subject to certain influence systems may lose its wisdom.  It is of
  great interest to characterize when an influence system preserves the
  crowd wisdom effect.

  In this paper we introduce and characterize numerous wisdom preservation
  properties of the basic French-DeGroot influence system model. Instead of
  requiring complete convergence to consensus as in the previous na\"ive
  learning model by Golub and Jackson, we study finite-time executions of
  the French-DeGroot influence process and establish in this novel context
  the notion of prominent families (as a group of individuals with outsize
  influence).  Surprisingly, finite-time wisdom preservation of the
  influence system is strictly distinct from its infinite-time version. We
  provide a comprehensive treatment of various finite-time wisdom
  preservation notions, counterexamples to meaningful conjectures, and a
  complete characterization of equal-neighbor influence systems.
\end{abstract}


\section{Introduction}

\textit{Problem description and motivation.} Social scientists describe as
\emph{wisdom of crowds} the phenomenon whereby large groups of people are
collectively smarter than single individuals. Much research has focused on
understanding the conditions under which such a phenomenon occurs.  It is
generally assumed that required conditions include diversity of opinions,
independence of opinions, and a system that aggregates individual opinions
without introducing bias.  Indeed, one plausible explanation for this
phenomenon in certain estimation tasks is that: (i) each individual opinion
is influenced by an independent zero-mean noise and, therefore, (ii) taking
averages of large numbers of individual opinions will reduce the effect of
noise.

A natural system that aggregates individual opinions is an influence
system, that is, a social dynamic process whereby individual opinions,
starting from their respective initial conditions, evolve possibly towards
consensus.  So it is natural to study the conditions under which an
influence system enhances, preserves, or diminishes the wisdom of crowd
effect. Motivated precisely by such reasoning, an insightful \emph{\naive
  learning model} was recently proposed by Golub and
Jackson~\cite{BG-MOJ:10}.  In this model, initial individual opinions are
measurements of an unknown parameter corrupted by independent zero-mean
noise. The law of large numbers then says that the crowd is initially
\emph{wise}, in the sense that the average initial opinion equals the
unknown parameter in the limit of infinitely many individuals.  Moreover,
and more relevant here, the social influence process is \emph{wisdom
  preserving} if, starting from wise individual estimates, the crowd
remains wise after the execution of the influence process. In other words,
the influence system is wisdom preserving if its execution does not
introduce bias in the crowd's estimate of the unknown parameter.

This paper is a contribution to the \naive learning model; instead of
requiring that the opinion formation process is allowed infinite time in
order for the crowd to reach consensus, we here investigate the more
realistic setting in which the influence system is executed only for finite
time. Accordingly, we investigate the \emph{finite-time wisdom preservation
  properties} of the influence system. For such finite-time settings, we
aim to define a number of wisdom notions, provide rigorous
characterizations, and work out instructive examples. In what follows, in
the interest of brevity and consistency with the original
work~\cite{BG-MOJ:10}, we refer to a wisdom preserving influence system as
a \emph{wise} influence system.

{\textit{Literature review.}}  Aristotle is thought to be the first author
to write about the ``wisdom of the crowd'' in his work ``Politics.''
Galton~\cite{FG:1907} describes a famous classic experiment concerning the
weight of a slaughtered ox and documents how the average estimate of the
ox's weight was indeed remarkably close to its correct value.  In recent
years, Surowiecki's book~\cite{JSU:04} has widely popularized the concept and discussed potential applications to group decision making
  and forecasting.

The literature on opinion dynamics and influence systems is very rich. The
classic and widely-established French-DeGroot averaging model is discussed
in the text on influence systems by Friedkin and Johnsen~\cite{NEF-ECJ:11},
the text on social networks by Jackson~\cite{MOJ:10}, and the recent survey
by Proskurnikov and~Tempo~\cite{AVP-RT:17}.  Inextricably linked with
influence systems is the concept of centrality in its various incarnations.
Especially relevant to this work is the notion of the eigenvector
centrality, e.g., see the foundational works by Bonacich~\cite{PB:72c} and
Friedkin~\cite{NEF:91}.  \cite{PMD-DV-JZ:03} provides a bounded rational
interpretation, a low-rank property, and the asymptotic behavior over
reducible graphs of the French-DeGroot averaging model.  The phenomenon of
wisdom of crowds is related also to the topic of social learning. We refer
to the original work~\cite{BG-MOJ:10} for a insightful comparison between
the \naive learning approach and other distinct methods. We here only
outline two alternative approaches. First, Acemoglu~\emph{et
  al}~\cite{DA-MAD-IL-AO:11} present a game-theoretical model for
sequential learning; even in this setup the presence of ``excessively
influential agents'' is an impediment to social learning. Second, the
social learning model proposed by Jadbabaie~\emph{et
  al}~\cite{AJ-AS-ATS:10} is based on the constant arrival of new
information (i.e., a feature our model does not have) and shows that new
information together with basic connectivity properties is sufficient to
overcome the influence of any individual agent.

The understanding of what social influence processes do indeed enhance
or diminishes the wisdom of a crowd continues to be of very recent
interest and debate.  Lorenz~\emph{et al}~\cite{JL-HR-FS-DH:11}
describe an empirical study in which the influence system diminishes
the accuracy of the collective estimate by diminishing diversity of
the crowd (even while conducing the opinions closer to consensus). In
contrast, Becker~\emph{et al}~\cite{JB-DB-DC:17} present theoretical
and empirical evidence about settings in which the influence system
enhances the wisdom of the crowd.
                
{\textit{Contributions.}}
We start by reviewing the \naive learning model for large populations
proposed by Golub and Jackson~\cite{BG-MOJ:10} and based on the
French-DeGroot opinion formation process. In this model, the
population is wise if the final (in the limit as time $k\to\infty$)
consensus opinion is equal to the correct value.  This property is
cast in terms of the left dominant eigenvector of the influence matrix
in large populations (i.e., in the limit as the number of individual
$n\to\infty$).  Specifically, a wise population is a sequences of
row-stochastic matrices whose left dominant eigenvector satisfies a
``vanishing maximum entry condition.''  This condition ensures that no
individual's initial erroneous opinion has an outside impact on the
final consensus opinion, thereby rendering the population unwise.

In this paper we consider a variation of the \naive learning setting in which
the French-DeGroot opinion formation process is allowed only finite time and
does not, therefore, reach completion. Our individuals do not reach consensus
and, in this paper, a population is finite-time wise if the average of the
individuals opinion remains correct as time progresses along the opinion
dynamics process.  In other words, \cite{BG-MOJ:10} considers wisdom in the
limit in which $n\to\infty$ after $k\to\infty$, we here consider the limit in
which $n\to\infty$ after at $k=1$, $k$ fixed, and uniformly over $k$.  We argue
that these finite-time variations are especially relevant for large population,
since it is known that the time required for consensus to be approximately
achieved typically diverges as the population size diverges. { Note
  that even controlled experiments are executed with a finite time constraint,
  e.g., see the recent work in~\cite{JB-DB-DC:17} where $k$ is precisely equal
  to 2.} It is therefore of interest to assume that no enough time is provided
to the opinion formation process to achieve complete consensus and inquire what
are wise sequences under this relaxed condition.

Given this premise, this paper makes four main contributions.  First,
for our proposed finite-time setting we introduce the notions of
one-time wisdom, finite-time wisdom, uniform wisdom, and pre-uniform
wisdom.  We provide necessary and sufficient characterizations of
one-time and finite-time wisdom in terms of the limiting value of the
maximum column averages of the matrix sequence. We also provide a
sufficient condition for uniform wisdom involving the same limiting
value as well as the limiting value of the matrices' mixing time.

Second, we provide numerous detailed examples of graph families to
establish, among other relationships, that finite-time wisdom neither
implies nor is implied by wisdom, as previously defined. Our examples
illustrate how various general implications among the various wisdom
notions are tight. Our examples also demonstrate that the proposed
notions of one-time and finite-time wisdom are meaningful and strictly
distinct from the notion of (infinite-time) wisdom as defined by Golub
and Jackson~\cite{BG-MOJ:10}.

Third, we provide a sufficient condition to ensure that a sequence of
row-stochastic matrices is indeed finite-time wise. We introduce an
appropriate novel notion of prominent family and show how its absence
implies finite-time wisdom in general. Roughly speaking, a collection
of nodes (a family) is prominent if, in the limit of large population,
its size is negligible (that is, of order $o(n)$) but its total
accorded one-time influence is not (that is, of order~1 in $n$).

Fourth, we then extend our sufficient condition to the setting of
equal-neighbor sequences of row-stochastic matrices.  For such
highly-structured case, we show that the absence of prominent
individuals is a necessary and sufficient condition for finite-time
wisdom, pre-uniform wisdom, and wisdom in the equal-neighbor setting.
In other words, we completely characterize finite-time wisdom in the
equal-neighbor setting. We apply these findings to the case of
preferential attachment model and show how (1) the classic
\Barabasi-Albert model (with attachment probability linearly
proportional to degree) is wise in all possible senses (finite-time,
uniformly and infinite-time), whereas (2) a super-linear preferential
attachment model loses all its wisdom properties.

{\textit{Paper organization.}}
Section~\ref{sec:definition+basicresults} contains various wisdom
definitions and corresponding general conditions.
Section~\ref{sec:examples} contains the analysis of several
counterexamples showing the general difference among the various
wisdom notions.  Section~\ref{sec:prominent-families} contains a
sufficient condition for finite-time wisdom and
Section~\ref{sec:iff-equal-neighbor} characterizes pre-uniform wisdom
in sequences of equal-neighbor matrices.
Section~\ref{sec:conclusions} contains some concluding remarks.

\subsection{Review of notation and preliminary concepts}
Let $\vect{1}_n\in \real^{n\times 1}$ denote the vector of all ones.  Given
$x\in\real^n$, we define its \emph{average} by $
\ave(x)=\frac{1}{n}\vect{1}_n^\top x=\frac{1}{n}\sum_{i=1}^nx_i$.  Given
$x\in\real^n$, its \emph{norms} by $\|x\|_1 = \sum_i |x_i|$, $\|x\|_2 =
(\sum_i x_i^2)^{1/2}$ and $\|x\|_\infty = \max_i|x_i|$.  Let
$\simplex{n}=\setdef{x\in\real^n}{x\geq0, \vect{1}_n^\top{x}=1}$.  Given
$x\in\simplex{n}$, the following inequalities hold
\begin{equation} \label{ineq:norms-on-simplex}
  \|x\|_{\infty}^2 \leq x^\top x \leq \|x\|_{\infty}.
\end{equation}
Indeed, the left-hand side is self evident, while right-hand side follows
from $x^\top x\leq \|x\|_{\infty}\|x\|_1=\|x\|_{\infty}$.

{\textit{Stochastic matrices.}}
A matrix $P\in\real^{n\times{n}}$ is non-negative ($P\geq0$) if
$P_{ij}\geq 0$ for each $i, j$. For such a matrix, the maximum column
sum is
\begin{equation*}
  \onenorm{P} = \max_{j\in\until{n}} \sum_{i=1}^n P_{ij} =
  \max \vect{1}_n^\top P.
\end{equation*}
Note that $\onenorm{\frac{1}{n}P}$ is the maximum column average of
$P$.  To a nonnegative matrix $P$ we associate an unweighted directed
graph $G=(\until{n},E)$ where $E=\setdef{(i,j)\in
  \until{n}^2}{P_{ij}>0}$. Given a node $j$, let $N_j$ denote the set
of out-neighbors of $j$.  The nonnegative matrix $P$ is
\emph{irreducible} if $G$ is strongly connected and \emph{primitive}
if $G$ is strongly connected and aperiodic.

A matrix $P\in\real^{n\times{n}}$ is \emph{stochastic} if $P\geq0$ and
$P\vect{1}_n=\vect{1}_n$.

Given a stochastic irreducible matrix $P$, its \emph{left dominant
  eigenvector} $\pi\in\simplex{n}$ is unique and well defined by the
Perron-Frobenius Theorem and satisfies $\pi^\top P = \pi^\top$.

Given a primitive stochastic matrix $P$, its
\emph{mixing time} is defined by
\begin{equation}
  \label{def:mixingtime}
  \taumix(P) := \inf\Bigsetdef{t\in\natural}{
    \max_{i,j}\sum\nolimits_k|(P)^t_{ik}-P)^t_{jk}|\le\frac{1}{\e}  }\,.
\end{equation}

{\textit{Equal-neighbor models.}}

\begin{definition}[Equal-neighbor, directed equal-neighbor,
    and weighted-neighbor matrices]
  \label{def:equal-neighbor-models}
  Given a nonnegative matrix $W\in\realnonnegative^{n\times{n}}$ with at
  least one positive entry in each row (possibly on the diagonal), define
  the stochastic matrix $P\in\real^{n\times{n}}$ by
  \begin{equation}\label{def:fromWtoP}
  P=\diag(W\vect{1}_n)^{-1}W.
  \end{equation}
  Then the matrix $P$ is said to be
  \begin{enumerate}
  \item \emph{equal-neighbor} if all non-zero entries of $W$ have the same
    value and $W=W^\top$,
  \item \emph{directed equal-neighbor} if all non-zero entries of $W$ have
    the same value and $W\neq W^\top$, and
  \item \emph{weighted-neighbor} if the non-zero entries of $W$ take more
    than one value and $W=W^\top$.
\end{enumerate}
  In these three definitions, the graph associated to $W$ is undirected
  unweighted, directed unweighted, and undirected weighted respectively.
\end{definition}

Note that, given a non-negative weight matrix
$W\in\realnonnegative^{n\times{n}}$, the \emph{weighted out-degree} of
a node $i$ in the weighted digraph associated to $W$ is $d_i=
(W\vect{1}_n)_i = \sum_{j=1}^n W_{ij}$.  Because we assume $d_i>0$ for
all $i$, equation~\eqref{def:fromWtoP} is well-posed and equivalent
to:
\begin{equation*} 
    P_{ij}= d_i^{-1} W_{ij} .
\end{equation*}
If $W$ is symmetric, then $d_i$ is called the \emph{weighted degree}.
If additionally $W$ is binary, then $d_i$ is called the \emph{degree}
and
\begin{equation}
  \label{def:Pij=diinv-Wij}
  P_{ij}= d_i^{-1} W_{ij} = \begin{cases} d_i^{-1}, \qquad &\text{if
    }j\in N_i, \\ 0, &\text{otherwise }.
  \end{cases}
\end{equation}

We conclude with a known useful result. The left dominant eigenvector $\pi$
of an irreducible weighted-neighbor matrix $P$ satisfies, for all
$i\in\until{n}$,
\begin{equation}
  \label{eq:pi-equal-neighbor}
  \pi_i = \frac{d_i}{d_1+\dots+d_n}.
\end{equation}

\section{Wisdom in sequences of stochastic matrices: definitions and basic results}
\label{sec:definition+basicresults}
We consider the classic French-DeGroot model~\cite{JRPF:56,MHDG:74} of
opinion dynamics
\begin{equation}
  \label{def:DeGroot}
  x_i(k+1)=\sum_{j=1}^nP_{ij}x_j(k),
\end{equation}
where $x_i(k)$ denotes the opinion of individual $i$, $i\in\until{n}$
at time $k\in\integernonnegative$.  The coefficient $P_{ij}$ denotes
the weight that individual $i$ assigns to individual $j$ when carrying
out this revision. The matrix $P$ is assumed row-stochastic and
defines a weighted directed graph $G$.

We assume
\begin{equation} \label{eq:truth}
  x_i(0)=\mu+\xi_i(0),
\end{equation}
where the constant $\mu\in\real$ is unknown parameter and the noisy
terms $\xi_i(0)$ are a family of independent Gaussian-distributed
variables such that $\EE[\xi_i(0)]=0$ and
$\Var[\xi_i(0)]=\sigma^2<\infty$ for all $i\in\{1,\dots,n\}$.

We will consider sequences of DeGroot models~\eqref{def:DeGroot} with
increasing dimensions $n$ and denote all relative quantities with a
superscript $[n]$. In particular, the state of the $n$-dimensional
model is $\supscr{x}{n}$.  Given assumption~\eqref{eq:truth} and
assuming $\mu$ is kept constant, the law of large numbers implies
that, almost surely,
\begin{equation*}
  \lim_{n\to\infty} \ave(\supscr{x}{n}(0))) = \mu.
\end{equation*}

\begin{definition}[Wisdom notions]
  \label{def:wisdom-notions}
  Given a \emph{sequence of stochastic matrices of increasing
    dimensions} $\{P^{[n]}\in\real^{n\times{n}}\}_{n\in\natural}$,
  define a sequence of opinion dynamics problems with initial state
  $\{\supscr{x}{n}(0)\in\real^n\}_{n\in\natural}$ satisfying
  assumption~\eqref{eq:truth} and evolution
  $\{\supscr{x}{n}(k)\in\real^n\}_{n\in\natural}$ at times
  $k\in\natural$.  The sequence
  $\{P^{[n]}\in\real^{n\times{n}}\}_{n\in\natural}$, is
  \begin{enumerate}
  \item\label{def:wise:1} \emph{one-time wise} if
    $\lim\limits_{n\to\infty} \ave(\supscr{x}{n}(1))) = \mu$;
  \item\label{def:wise:2} \emph{finite-time wise} if
    $\lim\limits_{n\to\infty} \ave(\supscr{x}{n}(k))) = \mu$, for all
    $k\in\natural$;
  \item\label{def:wise:4} \emph{wise} if $\lim\limits_{n\to\infty}
    \lim\limits_{k\to\infty} \ave(\supscr{x}{n}(k))) = \mu $;
  \item\label{def:wise:3} \emph{uniformly wise} if
    $\lim\limits_{n\to\infty} \sup\limits_{k\in\natural}
    |\ave(\supscr{x}{n}(k)))-\mu|=0$.
  \end{enumerate}
  Here all limits are meant in the probability sense.
\end{definition}

{ The terminology ``wisdom preserving'' instead of ``wise''
  would be more appropriate in this context. Indeed, we here study whether
  the dynamics driven by the matrices $P^{[n]}$ preserves or diminishes the
  crowd's wisdom assumed to hold at initial time. As already remarked in
  the Introduction, we have chosen this simpler name in the interest of
  brevity and consistency with the original work~\cite{BG-MOJ:10}.}

We note two straightforward relations among these notions of wisdom, that
is, \ref{def:wise:2} implies~\ref{def:wise:1} and~\ref{def:wise:3}
implies~\ref{def:wise:4}. More details on such implications are given in
Section~\ref{sec:examples}.

The following theorem provides necessary and sufficient characterizations
of the notions~\ref{def:wise:1}, \ref{def:wise:2}, and~\ref{def:wise:4} of
the above definition. The proof of these characterizations is immediate
from~\cite{BG-MOJ:10,WEP:66}.

\begin{theorem}[Necessary and sufficient conditions for finite-time wisdom and wisdom]
  \label{theo:wisdom}
  Consider a \emph{sequence of stochastic matrices of increasing
    dimensions} $\{P^{[n]}\in\real^{n\times{n}}\}_{n\in\natural}$.
  The sequence is
  \begin{enumerate}
  \item\label{wc1} \emph{one-time wise} if and only if
    $\lim\limits_{n\to\infty} \onenorm{\frac{1}{n} P^{[n]}} =0$;
  \item\label{wc2} \emph{finite-time wise} if and only if, for all
    $k\in\natural$, $\lim\limits_{n\to\infty} \onenorm{\frac{1}{n}
    (P^{[n]})^k} =0$. \setcounter{saveenum}{\value{enumi}}
  \end{enumerate}
  Moreover, assuming $\{P^{[n]}\}_{n\in\natural}$ are primitive, the
  sequence is
  \begin{enumerate}\setcounter{enumi}{\value{saveenum}}
  \item\label{wc4} \emph{wise} if and only if
    $\lim\limits_{n\to\infty} \inftynorm{\pi^{[n]}} =0$, where
    $\pi^{[n]}\in\simplex{n}$ is the left dominant eigenvector of
    $P^{[n]}$, for $n\in\natural$.
  \end{enumerate}
\end{theorem}
  
\begin{proof}
  Define
  \begin{equation}\label{def:flux}
    \chi^{[n]}(k)=\frac{1}{n}\vect{1}_n^\top (P^{[n]})^k
  \end{equation}
  and notice that
$$\ave(\supscr{x}{n}(k))=\mu+(\chi^{[n]}(k))^{\top}\xi^{[n]}(0).$$
Thus,
\begin{align}\label{unbiased}
  \EE[\ave(\supscr{x}{n}(k))]&=\mu,\\
  \Var[\ave(\supscr{x}{n}(k))]&=\sigma^2 (\chi^{[n]}(k))^\top \chi^{[n]}(k).
\end{align}
By Chebyshev's inequality, the convergence in probability of
$\ave(\supscr{x}{n}(k))$ to $\mu$ is equivalent to the convergence to $0$
of the variance.  Statements~\ref{wc1} and \ref{wc2} can now be proven by
applying the inequalities~\eqref{ineq:norms-on-simplex} with
$x=\chi^{[n]}(k)$ and noting that
$\|\chi^{[n]}(k)\|_{\infty}=\onenorm{\frac{1}{n} (P^{[n]})^k}$.

Regarding statement~\ref{wc4}, the primitivity assumption implies
$\lim\limits_{k\to +\infty}\supscr{x}{n}(k)= (\pi^{[n]})^{\top}
\big(\mu\vect{1}_n+\xi^{[n]}(0)\big)\vect{1}_n$ where $\pi^{[n]}$ is the
left dominant eigenvector of $\supscr{P}{n}$. Therefore,
\begin{equation*}
  \lim\limits_{k\to +\infty}\ave(\supscr{x}{n}(k))=\mu+(\pi^{[n]})^{\top}\xi^{[n]}(0),
\end{equation*}
and statement~\ref{wc4} follows by applying the
inequalities~\eqref{ineq:norms-on-simplex} with $x=\pi^{[n]}$.
\end{proof}

\begin{remark}[Estimation theory interpretation]
  Adopting statistics language, our paper is concerned with the properties of
  the arithmetic average estimator $x\mapsto\sum_{i=1}^nx_i/n$ applied to the
  stochastic process $\{x^{[n]}(k)\}_{k\in\natural}$. Equation~\eqref{unbiased}
  implies that, for every $n$ and for every $k$, the arithmetic average
  estimator, regarded as a random variable, is \emph{unbiased}. We are
  interested in \emph{concentration results}, that is, we investigate when, as
  $k$ varies and as $n\to+\infty$, the estimator has its variance converging to
  $0$, namely, when the estimator concentrates around its mean value.
\end{remark}


There is a natural property that, in analogy to the other characterizations
presented in Theorem~\ref{theo:wisdom}, is related to uniform wisdom. We present
it as a separated concept.

\begin{definition}[Pre-uniform wisdom]
  \label{def:pre-uniform-wisdom}A sequence of stochastic matrices of increasing dimensions
 $\{P^{[n]}\in\real^{n\times{n}}\}_{n\in\natural}$ is
 \emph{pre-uniformly wise} if $\lim\limits_{n\to\infty}
 \sup\limits_{k\in\natural}\onenorm{\frac{1}{n}(P^{[n]})^k} =0$.
 \end{definition}

While pre-uniform wisdom is not sufficient to guarantee uniform
wisdom, it surely yields wisdom and finite-time wisdom, as an
immediate consequence of Theorem~\ref{theo:wisdom}.

Finally, for what concerns uniform wisdom, we present a sufficient
condition that turns out to be useful in many applications. Recall the
notion of mixing time from equation~\eqref{def:mixingtime}.

\begin{theorem}[A sufficient condition for uniform wisdom]
  \label{theo:sufficient-uniform}
 A \emph{sequence of primitive stochastic matrices of increasing
   dimensions} $\{P^{[n]}\in\real^{n\times{n}}\}_{n\in\natural}$ is
 uniformly wise if
  \begin{equation*}
    \lim\limits_{n\to+\infty}
    \sup\limits_{k\in\natural}\Bigonenorm{\frac{1}{n}(P^{[n]})^k}
    \taumix(P^{[n]})=0.
  \end{equation*}
\end{theorem}

\begin{proof}
Define
\begin{equation*}
  Y^{[n]}(k)=\ave(\supscr{x}{n}(k)))-\mu = \chi^{[n]}(k)^{\top}\xi^{[n]}(0),
\end{equation*}
where $\chi^{[n]}(k)$ is defined in \eqref{def:flux}. Notice that
\begin{equation*}
  \lim\limits_{k\to +\infty}Y^{[n]}(k)=Y^{[n]}:=(\pi^{[n]})^{\top}\xi^{[n]}(0).
\end{equation*}
Fix $\delta >0$ and notice that
\begin{multline}
  \label{eq:break}
  \Prob\!\Big[\sup\limits_{k\in\N} |Y^{[n]}(k)|>\delta\Big] \leq
  \Prob\!\Big[\sup\limits_{k\in\N} |Y^{[n]}(k)- Y^{[n]}|>\delta/2\Big]+ \Prob[|Y^{[n]}|>\delta/2].
\end{multline}
We now estimate the two right-end-side terms of
equation~\eqref{eq:break}.  Regarding the first term, we fix $k$ and
we use Chebyshev's inequality to obtain
\begin{align}
  \Prob[| Y^{[n]}(k)- Y^{[n]}|>\delta/2]
  &\leq   \frac{4}{\delta^2}\Var[Y^{[n]}(k)- Y^{[n]}]
  \nonumber \\
  &\leq \frac{4}{\delta^2}(\chi^{[n]}(k)-\pi^{[n]})^{\top}(\chi^{[n]}(k)-\pi^{[n]})
  \nonumber \\
  &\leq \frac{4}{\delta^2}\|\chi^{[n]}(k)-\pi^{[n]}\|_{\infty}\|\chi^{[n]}(k)-\pi^{[n]}\|_{1}
  \label{eq:estim1}
\end{align}
Since $\|\chi^{[n]}(k)\|_{\infty}=\onenorm{\frac{1}{n}
    (P^{[n]})^k}$ and $\chi^{[n]}(k)\to\pi^{[n]}$ for $k\to +\infty$, the
  $\infty$-norm factor can be estimated as
\begin{multline}\label{eq:estim3a}
\|\chi^{[n]}(k)-\pi^{[n]}\|_{\infty}\leq \|\chi^{[n]}(k)\|_{\infty}+\|\pi^{[n]}\|_{\infty} \leq 2\sup\limits_{k\in\natural}
\Bigonenorm{\frac{1}{n}(P^{[n]})^k}.
\end{multline}
The final $1$-norm factor in equation~\eqref{eq:estim1} can be estimated in
terms of the mixing time~\cite[Section 4.5]{DAL-YP-ELW:09}
\begin{equation}\label{eq:estim3}
  \|\chi^{[n]}(k)-\pi^{[n]}\|_{1}\leq \exp\left(-\left\lfloor\frac{k}{\taumix(P^{[n]})}\right\rfloor\right).
\end{equation}
Substituting (\ref{eq:estim3a}) and (\ref{eq:estim3}) into
inequality~\eqref{eq:estim1} and using a union bound estimation (also
referred to as Boole's inequality), we now obtain
\begin{align}
  \Prob\!\Big[\sup\limits_{k\in\N}|Y^{[n]}(k)- Y^{[n]}|>\delta/2\Big] 
 & \leq\frac{8}{\delta^2} \sup\limits_{k\in\natural}\Bigonenorm{\frac{1}{n}(P^{[n]})^k}\sum_k \exp\left(-\left\lfloor\frac{k}{\taumix(P^{[n]})}\right\rfloor\right)
  \nonumber \\
  &=\frac{8e}{\delta^2} \sup\limits_{k\in\natural}\Bigonenorm{\frac{1}{n}(P^{[n]})^k}\frac{1}{1- \exp\left(-\frac{1}{\taumix(P^{[n]})}\right)}
  \nonumber \\
  &\leq \frac{8e}{\delta^2} \sup\limits_{k\in\natural}\Bigonenorm{\frac{1}{n}(P^{[n]})^k}\taumix(P^{[n]}). \label{eq:estim4}
\end{align}
The second term in the right-hand-side of equation~\eqref{eq:break}
can be easily bounded using Chebyshev's inequality and the simple
bound~\eqref{ineq:norms-on-simplex}:
\begin{align}
  \Prob[|Y^{[n]}|>\delta/2]
  &\leq\frac{4}{\delta^2}\Var[Y^{[n]}] = \frac{4}{\delta^2}(\pi^{[n]})^{\top}(\pi^{[n]})
  \nonumber \\
  &\leq\frac{4}{\delta^2}\|\pi^{[n]}\|_{\infty}
  \nonumber \\
  &\leq\frac{4}{\delta^2}\sup\limits_{k\in\natural}\Bigonenorm{\frac{1}{n}(P^{[n]})^k}.
  \label{eq:estim2}
\end{align}
The theorem statement follows from putting together the
inequalities~\eqref{eq:break}, \eqref{eq:estim4}, and
\eqref{eq:estim2}.
\end{proof}


\section{Basic implications and counterexamples}
\label{sec:examples}

In the following lemma we show how the basic wisdom notions in
Definitions~\ref{def:wisdom-notions} and~\ref{def:pre-uniform-wisdom}
are related by a simple implication and how four counterexamples show
that no additional statements may be made in general. To construct our
counterexamples, we will rely upon the equal neighbor models as in
Definition~\ref{def:equal-neighbor-models}.

\begin{lemma}[Basic implications and counterexamples]
  \label{lemma:basic-implications}
  If a sequence of primitive stochastic matrices of increasing
  dimensions $\{P^{[n]}\in\real^{n\times{n}}\}_{n\in\natural}$ is uniformly wise or
  pre-uniformly wise, then it is wise and finite-time wise.  Moreover,
  there exist sequences that
  \begin{enumerate}
  \item are neither wise nor finite-time wise (see
    Example~\ref{ex:star}),
  \item are wise, but not one-time wise (see
    Example~\ref{ex:complete+star}),
  \item are finite-time wise, but not wise (see
    Example~\ref{ex:biased-pathgraph}), and
  \item are wise and finite-time wise, but not pre-uniformly wise (see
    Example~\ref{ex:binary-tree}).
\end{enumerate}
\end{lemma}

We illustrate this lemma in Figure~\ref{fig:wisdom-sets}.  The proof
of the first statement in the lemma is an immediate consequence of
Definition~\ref{def:pre-uniform-wisdom} and the characterizations in
Theorem~\ref{theo:wisdom}.  The four existence statement are
established by providing explicit examples here below.

Each example includes a figure illustrating a numerical simulation of the
corresponding dynamics. In the simulations the initial opinions are
independently normally distributed with zero mean and unit standard
deviation, except for $\enne{x}_1(0)=1$ (for illustration purposes).

\begin{figure*}[bth]
  \begin{center}
    \includegraphics[width=\textwidth]{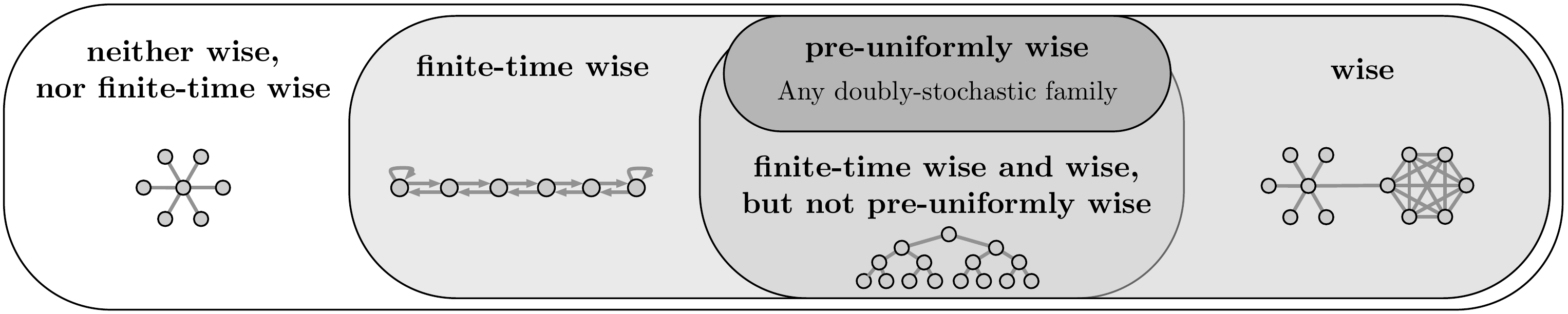}
  \end{center}
  \caption{Logical relations among sets of sequences with varying
    degree of wisdom. For each set we provide an example sequence that
    belongs to that set, but not to any strict subset of it.}
  \label{fig:wisdom-sets}
\end{figure*}


\begin{example}[The star graph is neither wise nor finite-time wise]
  \label{ex:star} For $n\in\natural$, define $\{P^{[n]}\in\real^{n\times{n}}\}_{n\in\natural}$ by 
  \begin{gather*}
    \supscr{P}{n} = \begin{bmatrix} 1/n & 1/n & 1/n & \cdots &
      1/n \\ 1 & 0 & 0 & \vdots & 0 \\ 1 & 0 & 0 & \vdots & 0
      \\ \vdots & \vdots & \vdots & \ddots & \vdots \\ 1 & 0 & 0 &
      \cdots & 0
    \end{bmatrix} .
  \end{gather*}
  In other words, consider the sequence of primitive equal-neighbor
  matrices defined by increasing-dimension star graphs (with a
  self-loop at the central node).  It is easy to see that the dominant
  left eigenvector of $\supscr{P}{n}$ is $\supscr{\pi}{n}
  = \begin{bmatrix} \frac{n}{2n-1} & \frac{1}{(2n-1)} & \dots &
    \frac{1}{(2n-1)}
  \end{bmatrix}^\top$.
  The sequence $\{P^{[n]}\}_{n\in\natural}$ is neither wise, nor
  one-time wise because
  \begin{align*}
    \lim_{n\to\infty} \inftynorm{\pi^{[n]}}  &=   \lim_{n\to\infty} \pi_1^{[n]}  = \frac{1}{2} , \\
    \lim_{n\to\infty} \onenorm{\tfrac{1}{n} P^{[n]}}  &=
    \lim_{n\to\infty} \max_{j\in\until{n}} \frac{1}{n} \sum_{i=1}^{n} \supscr{P}{n}_{ij} \\
    & =  \lim_{n\to\infty}  \frac{1}{n} \sum_{i=1}^{n} \supscr{P}{n}_{i1} 
    =  \lim_{n\to\infty} \frac{n-1+1/n}{n} = 1. 
  \end{align*}
  The lack of wisdom and finite-time wisdom is illustrated in Figure~\ref{fig:stella}.
  \oprocend
  \begin{figure}[H]\centering
    \includegraphics[width=.6\linewidth]{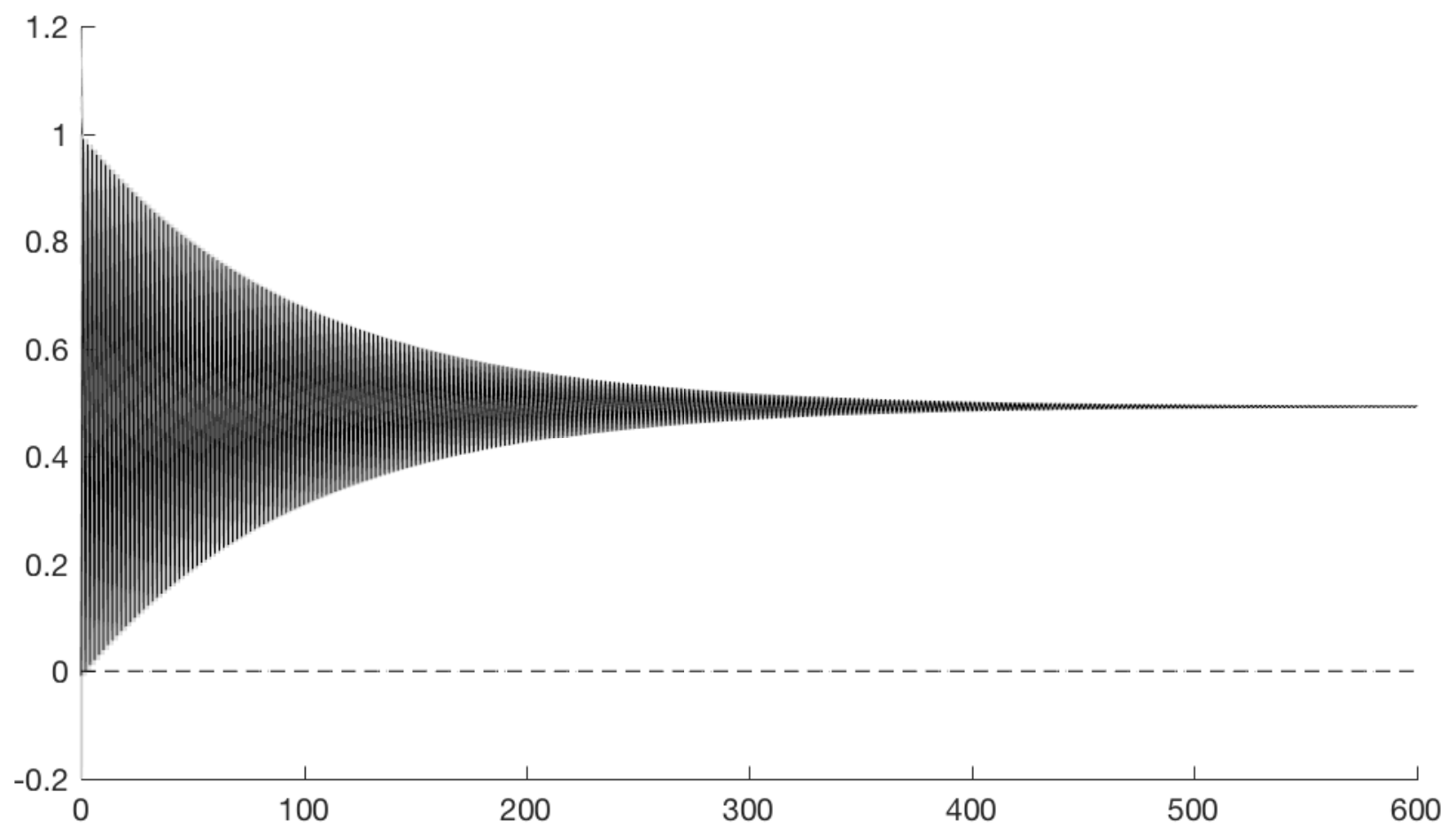}
    \caption{Star graph with $n=100$ nodes (one root and $99$ leaves). The
      time horizon is $T=600$. Individual opinions are in gray, the average
      opinion is in black. As proved in Example~\ref{ex:star}, the figure
      illustrates how the graph is neither wise nor finite-time
      wise.}\label{fig:stella}
  \end{figure}  
\end{example}

\begin{example}[The union/contraction of star and complete graph is wise, but not one-time wise]
  \label{ex:complete+star} For $n\in\natural$, let 
  $S_n\sqcup K_n$ denote the undirected graph with $2n$ nodes obtained
  by (i) computing the union of the star $S_n$ (one center node and
  $n$ leafs) and a complete graph $K_n$, and (ii)
  identifying/contracting one leaf of $S_n$ with a node of $K_n$.
  Accordingly, let $\{\supscr{W}{n}\}_{n\in\natural}$ be the sequence
  of adjacency matrices of $\{S_n\sqcup K_n\}_{n\in\natural}$ and
  $\{P^{[n]}\}_{n\in\natural}$ be the corresponding sequence of
  equal-neighbor matrices. These matrices are primitive because $K_n$
  contains cycles with co-prime length. Note the slight abuse of
  notation: $P^{[n]}$ has dimension $2n$.

  The graph $S_n\sqcup K_n$, for $n=6$, is depicted in
  Figure~\ref{fig:graph-star+complete}, where nodes are numbered as
  follows: node $1$ is always the center of the star and node $n+1$ is
  the node belonging to both the star and the complete graph.
  \begin{figure}[H]
    \begin{center}
      \includegraphics[width=.6\linewidth]{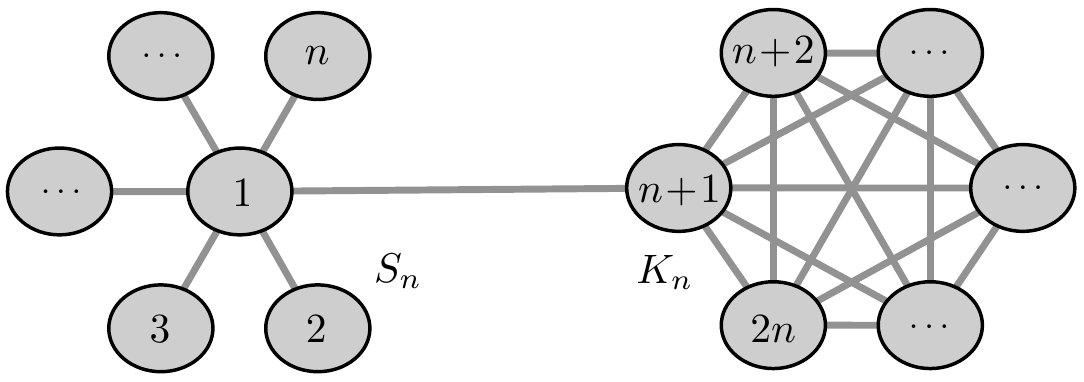}
    \end{center}
    \caption{The union and contraction of a star graph with a complete
      graph.}
    \label{fig:graph-star+complete}
  \end{figure}

  For this graph $S_n\sqcup K_n$, from
  equation~\eqref{eq:pi-equal-neighbor}, we compute the left dominant
  eigenvector of $\supscr{P}{n}$. First, we observe that: $d_1=n$,
  $d_j=1$ for a generic leaf $j\in\fromto{2}{n}$, $d_{n+1}=n$, and
  $d_h=n-1$ a generic node $h\in\fromto{n+2}{2n}$ in $K_n$. Hence, the
  sum of all degrees is $n+1\cdot(n-1)+n+(n-1)^2=n^2+n$ and the left
  dominant eigenvector satisfies:
  \begin{gather*}
    \supscr{\pi}{n}_1 = \frac{n}{n^2+n},\enspace
    \supscr{\pi}{n}_j = \frac{1}{n^2+n},\\
    \supscr{\pi}{n}_{n+1} = \frac{n}{n^2+n},\enspace
    \supscr{\pi}{n}_h = \frac{n-1}{n^2+n}.
  \end{gather*}
  Next, we compute the column sums of $\supscr{P}{n}$.  For the two
  special nodes we have
  \begin{align*}
    \sum_{i=1}^{2n} \supscr{P}{n}_{i1} &= \sum_{i=2}^{n} \supscr{P}{n}_{i1}  + \supscr{P}{n}_{n+1,1} 
    = (n-1) + \frac{1}{n},\\
    \sum_{i=1}^{2n} \supscr{P}{n}_{i,n+1} &= 
    \supscr{P}{n}_{1,n+1} + \sum_{i=n+2}^{2n} \supscr{P}{n}_{i,n+1}
   = 1 + \frac{1}{n},
  \end{align*}
  and, for a generic leaf $j\in\fromto{2}{n}$ in $S_n$ and a generic
  $h\in\fromto{n+2}{2n}$ in $K_n$, we have
  \begin{align*}
    \sum_{i=1}^{2n} \supscr{P}{n}_{ij} &= \frac{1}{n}, \\
    \sum_{i=1}^{2n} \supscr{P}{n}_{ih} &=
    \supscr{P}{n}_{n+1,h}  + 
    \sum_{i=n+2}^{2n} \supscr{P}{n}_{ih}  = \frac{1}{n} + (n-2)\frac{1}{n-1}.
  \end{align*}

  In summary, we note that the sequence $\{P^{[n]}\}_{n\in\natural}$ is
  wise but not one-time wise because
  \begin{align*}
    \lim_{n\to\infty} \inftynorm{\pi^{[n]}}  &=   \lim_{n\to\infty} \pi_1^{[n]}  = \lim_{n\to\infty} \frac{n}{n^2+n} = 0, \\
    \lim_{n\to\infty} \onenorm{\tfrac{1}{2n} P^{[n]}}  &=
    \lim_{n\to\infty} \max_{j\in\until{2n}} \frac{1}{2n} \sum_{i=1}^{2n} \supscr{P}{n}_{ij} \\
    &=  \lim_{n\to\infty}  \frac{1}{2n} \sum_{i=1}^{2n} \supscr{P}{n}_{i1}
    =  \lim_{n\to\infty}  \frac{n-1+\frac{1}{n}}{2n} = \frac{1}{2}. 
  \end{align*}
  The lack of one-time wisdom and the presence of wisdom is illustrated in
  Figure~\ref{fig:stella+completo}. \oprocend
\begin{figure}[H]\centering
\includegraphics[width=.6\linewidth]{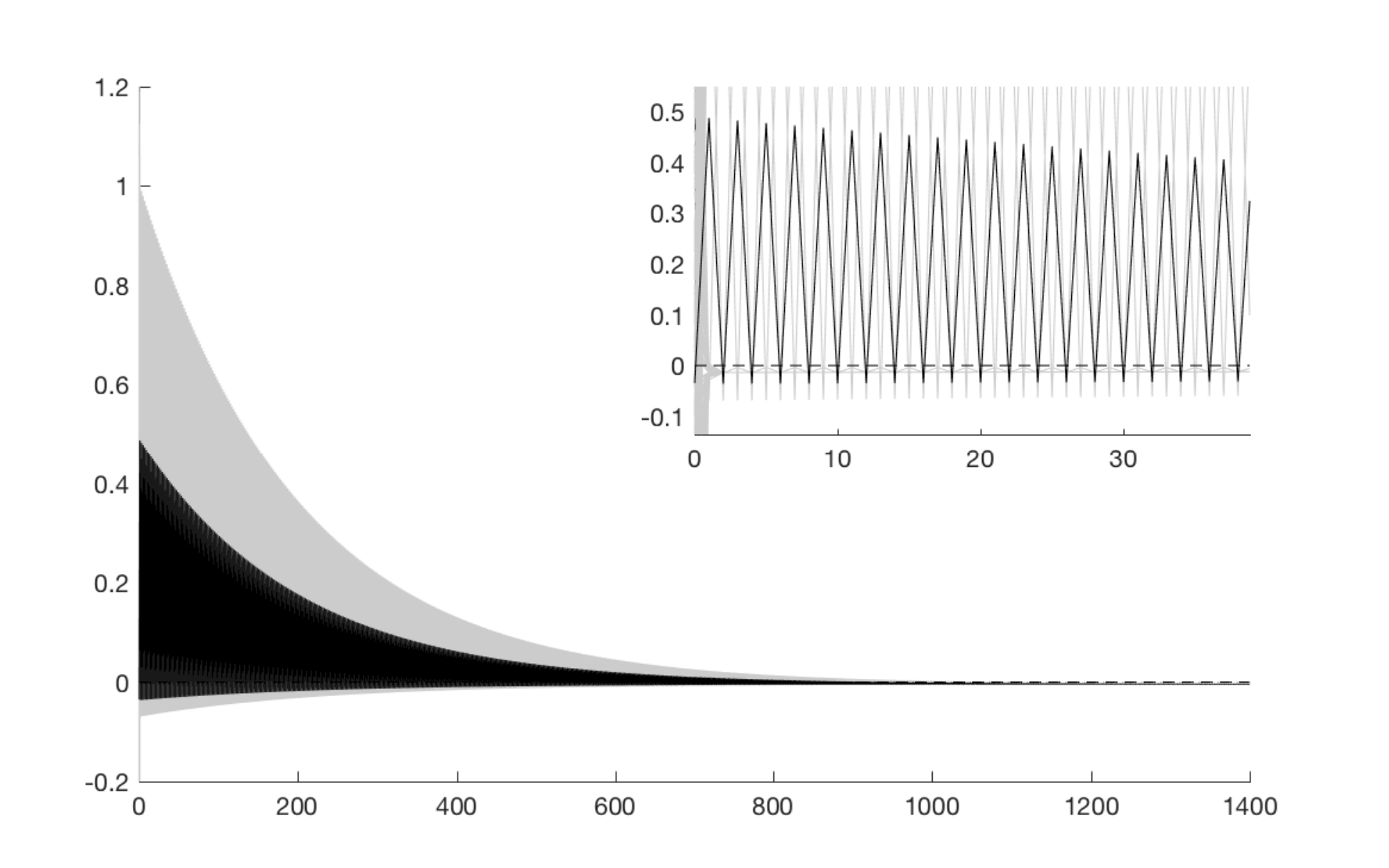}
\caption{Union/contraction of a star and a complete graph with $n=200$
  nodes. The time horizon is $T=1400$. Individual opinions are in gray, the
  average opinion is in black. As proved in Example~\ref{ex:complete+star},
  this graph is wise but not one-time wise.  The right upper-corner figure
  illustrates the very first steps of the evolution and shows that, for
  $k=1$, the average of the opinions moves away form
  zero.\label{fig:stella+completo}}
\end{figure}
\end{example}

\begin{example}[The path graph with biased weights is finite-time wise, but not wise]
  \label{ex:biased-pathgraph}
  Consider a sequence of increasing-dimension path graphs whose biased
  weights are selected as follows: (i) pick a constant $\nu>1$ and
  define the unique scalars $q>p>0$ satisfying $q/p=\nu$ and $p+q=1$,
  and (ii) define the $n$-dimensional weighted digraph as in
  Figure~\ref{fig:graph-directedpath} (self-loops at node $1$ and
  $n$).  Let $\{P^{[n]}\}_{n\in\natural}$ denote the sequence of
  primitive stochastic adjacency matrices of the the path graphs with
  biased weights. We have:
  \begin{figure}[H]
    \begin{center}
      \includegraphics[width=.7\linewidth]{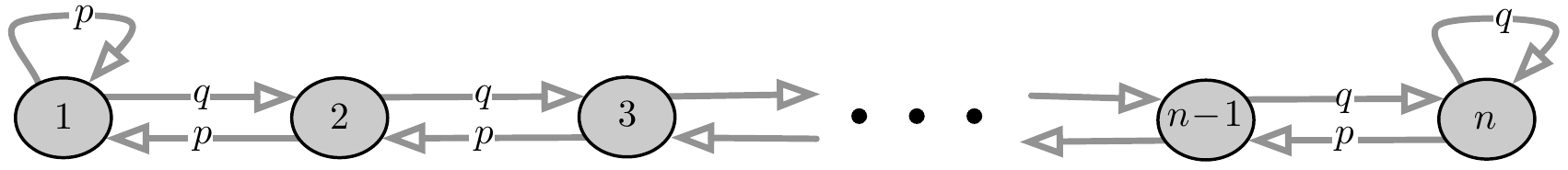}
    \end{center}
    \caption{Directed path with biased weights}
    \label{fig:graph-directedpath}
  \end{figure}
  \begin{gather*}
    \supscr{P}{n} = \begin{bmatrix} 
    p   & q              & \cdots   & 0      & 0  \\
    p      & 0       & q   & \ddots & 0      \\
    \vdots     & \ddots        & \ddots   & \ddots & 0      \\
    0          & \ddots     & p    & 0 & q  \\
    0      & 0        & \cdots        & p    & q
    \end{bmatrix}.    
  \end{gather*}
   Let $\supscr{\pi}{n}$ be the dominant left eigenvector of
   $\supscr{P}{n}$. We claim that
  \begin{equation}
    \label{eq:pi-biasedpath}
    \supscr{\pi}{n}_1 = \frac{\nu-1}{\nu^n-1}, \qquad \supscr{\pi}{n}_i = \nu^{i-1}\supscr{\pi}{n}_1,
    \text{ for $i\in\fromto{2}{n}$}.
  \end{equation}
  We prove this claim as follows. From $(\supscr{\pi}{n})^\top=(\supscr{\pi}{n})^\top \supscr{P}{n}$, we get
  \begin{align*}
     \supscr{\pi}{n}_1 p + \supscr{\pi}{n}_2 p &= \supscr{\pi}{n}_1, \\
     \supscr{\pi}{n}_{i-1} q + \supscr{\pi}{n}_{i+1} p &= \supscr{\pi}{n}_i, \quad \text{for } i\in\fromto{2}{n-1},\\
     \supscr{\pi}{n}_{n-1} q + \supscr{\pi}{n}_n q &= \supscr{\pi}{n}_n.
  \end{align*}
  From the first equality we immediately have $\supscr{\pi}{n}_2=\nu\supscr{\pi}{n}_1$. Next
  we prove by recursion that $\supscr{\pi}{n}_i=\nu^{i-1}\supscr{\pi}{n}_1$. This statement is
  true for $i=2$. Assuming it is true for arbitrary $i$, we compute
  \begin{align*}
    \supscr{\pi}{n}_{i+1} &= \frac{1}{p}\supscr{\pi}{n}_i-\frac{q}{p}\supscr{\pi}{n}_{i-1} =  \frac{1}{p}\nu^{i-1}\supscr{\pi}{n}_1-\frac{q}{p}\nu^{i-2}\supscr{\pi}{n}_1  \\
    &= \Big( \frac{1}{p}\nu -\frac{q}{p} \Big) \nu^{i-2} \supscr{\pi}{n}_1 =
    \frac{1-p}{p} \nu^{i-1} \supscr{\pi}{n}_1 = \frac{q}{p} \nu^{i-1} \supscr{\pi}{n}_1.
  \end{align*}
  The value of $\supscr{\pi}{n}_1$ in~\eqref{eq:pi-biasedpath} follows from
  requiring
  $1=\sum_{i=1}^n\supscr{\pi}{n}_i=\sum_{i=1}^n\nu^{i-1}\supscr{\pi}{n}_1=\frac{\nu^n-1}{\nu-1}\supscr{\pi}{n}_1$. This
  concludes our proof of the formula~\eqref{eq:pi-biasedpath} for
  $\supscr{\pi}{n}$.

  Note that from formula~\eqref{eq:pi-biasedpath}, we know
  $\supscr{\pi}{n}_n=\frac{\nu-1}{\nu^n-1}\nu^{n-1} \asymp
  \frac{\nu-1}{\nu} \text{ for }n\to\infty$.  In summary, we note that
  the sequence $\{P^{[n]}\}_{n\in\natural}$ is not wise but
  finite-time wise because
  \begin{align*}
    \lim_{n\to\infty} \inftynorm{\pi^{[n]}}  &=   \lim_{n\to\infty} \pi_n^{[n]}  = \frac{\nu-1}{\nu}> 0, \\
    \lim_{n\to\infty} \onenorm{\tfrac{1}{n} (P^{[n]})^k}  &\leq
    \lim_{n\to\infty} \onenorm{\tfrac{1}{n} P^{[n]}}^k  \\ &=
    \lim_{n\to\infty} \Big( \max_{j\in\until{n}} \frac{1}{n} \sum_{i=1}^{n} \supscr{P}{n}_{ij} \Big)^k\\
    &\leq  \lim_{n\to\infty}  \frac{2^k}{n}  = 0,  \quad \text{for all }k\in\natural, 
  \end{align*}
  where we used the fact that the in-degree of each node is upper bounded
  by $2$. The presence of finite time wisdom and the lack of wisdom are
  illustrated in Figure~\ref{fig:biased-pathgraph}.  \oprocend
  \begin{center}
    \begin{figure}[H]
      \centering
      \includegraphics[width=.6\linewidth]{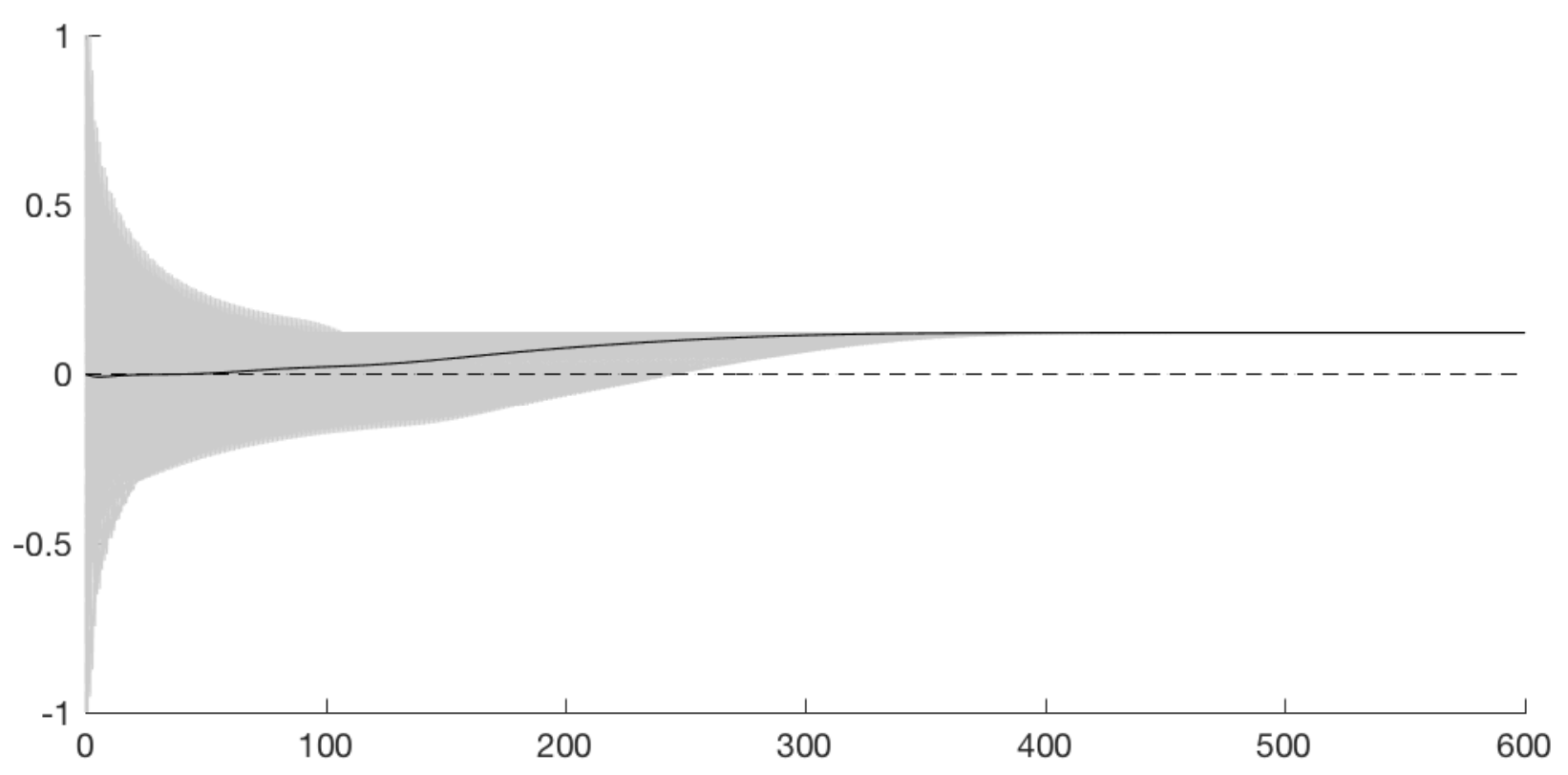}
      \caption{Path graph with $n=100$ nodes and weight parameters $p=1/3$
        and $q=2/3$. The time horizon is $T=600$. Individual opinions are
        in gray, the average opinion is in black. We observe
          that the average opinion is approximately zero for a duration of
          time proportional to the group size and then transitions to a
          steady-state error. As proved in
        Example~\ref{ex:biased-pathgraph}, the graph is finite-time wise
        but not wise.}\label{fig:biased-pathgraph}
    \end{figure}
  \end{center}    
\end{example}

\begin{example}[The reversed binary tree with root-leaves
    edges is wise and finite-time wise, but not pre-uniformly
    wise]\label{ex:binary-tree} We define a sequence of directed
  equal-neighbor model (see Definition~\ref{def:equal-neighbor-models}) by
  defining the corresponding sequence of binary (not symmetric)
  $\supscr{W}{n}$ matrices.  We consider a binary tree with $L$ layers and
  $n$ nodes, where $n=1+\dots+2^{L-1}=\frac{2^{L}-1}{2-1}=2^{L}-1$. As in
  Figure~\ref{fig:graph-heart}, we label the nodes as follows: $v_1^{(1)}$
  at layer $1$, $v_1^{(2)},v_2^{(2)}$ at layer $2$, and, more generally,
  $v_1^{(\ell)},\cdots,v_{2^{\ell-1}}^{(\ell)}$ at layer $\ell$, for
  $\ell\in\until{L}$. Note the self-loop at node $v_1^{(1)}$.
  \begin{figure}[H]
    \begin{center}
      \includegraphics[width=.7\linewidth]{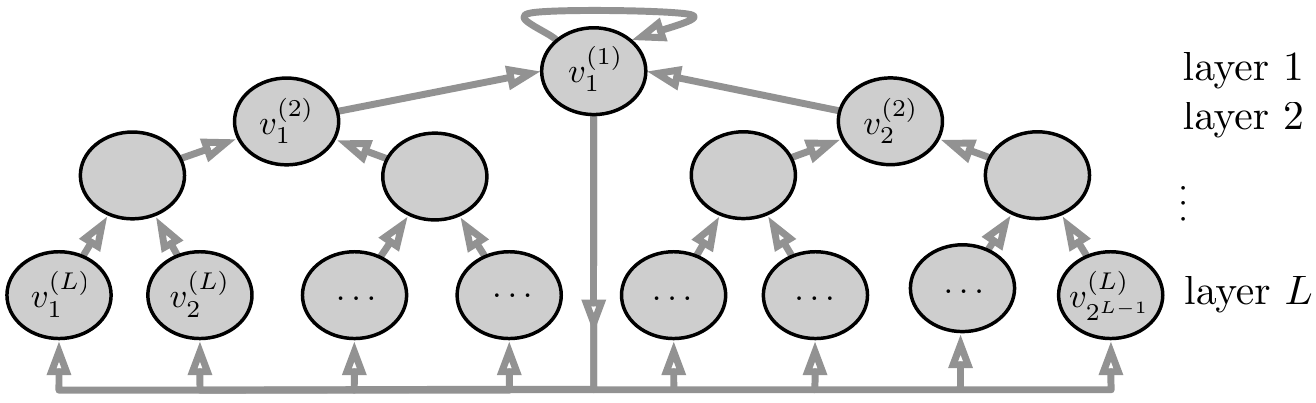}
    \end{center}
    \caption{A reversed binary directed tree with root-leaves edges.}
    \label{fig:graph-heart}
  \end{figure}

  
  We now compute the left dominant eigenvector $\supscr{\pi}{n}$.  Pick
  $\ell\in\until{L}$.  For symmetry reasons, we have
  $\supscr{\pi}{n}(v_1^{(\ell)}) = \cdots =
  \supscr{\pi}{n}(v_{2^{\ell-1}}^{(\ell)})$ so that, defining $\pi^{(\ell)}
  = \supscr{\pi}{n}(v_1^{(\ell)}) + \cdots +
  \supscr{\pi}{n}(v_{2^{\ell-1}}^{(\ell)})$, we compute $\pi^{(\ell)}
  =2^{\ell-1} \supscr{\pi}{n}(v_1^{(\ell)})$.
  An aggregation argument leads to $\pi^{(2)} = \cdots = \pi^{(L)}
  =\frac{1}{L+1}$ and $\pi^{(1)} =2\pi^{(2)}=\frac{2}{L+1}$.  Hence, in
  summary, for each node $h\in\until{2^{\ell-1}}$ at layer $\ell$,
  \begin{equation*}
    \supscr{\pi}{n}(v_h^{(\ell)}) =
    \begin{cases}
      \frac{2}{L+1}, \quad & \text{if } \ell=1,\\
      \frac{1}{2^{\ell-1}}\frac{1}{L+1}, \qquad & \text{if }  \ell>1.
    \end{cases}
  \end{equation*}  
  We can now state that the sequence is wise, since
  \begin{align*}
    \inftynorm{\pi^{[n]}}  &= 
    \max_{\ell\in\until{L}}
    \max_{h\in\until{2^{\ell-1}}}
    \supscr{\pi}{n}(v_h^{(\ell)}) \\ &=
    \supscr{\pi}{n}(v_1^{(1)}) = \frac{2}{L+1} = \frac{1}{\log_2(n+1)},
  \end{align*}
  and therefore $\lim_{n\to\infty} \inftynorm{\pi^{[n]}} =0$.

  Next, we note $\sum_{i=1}^n\supscr{P}{n}_{ij}\leq
  \sum_{i=1}^n\supscr{W}{n}_{ij} \leq 3$, where we used
  $\supscr{P}{n}_{ij}=d_i^{-1}\supscr{W}{n}_{ij}$, $d_i\geq1$, and
  knowledge of the fact that each node has at most $3$ in-edges. We can now
  state that the sequence is finite-time wise since, for all fixed time
  $k$,
  \begin{align*}
    \lim_{n\to\infty} \onenorm{\tfrac{1}{n} (P^{[n]})^k}  
    &\leq    \lim_{n\to\infty} \tfrac{1}{n} \onenorm{ P^{[n]}}^k  \\
    & =
    \lim_{n\to\infty} \tfrac{1}{n} \Big( \max_{j\in\until{n}} \sum_{i=1}^{n} \supscr{P}{n}_{ij} \Big)^k
    \\
    &\leq  \lim_{n\to\infty}  \frac{3^k}{n}  = 0.
  \end{align*}

  Finally, the sequence is not pre-uniformly wise because, when
  $k=k(n)=L-1=\log_2(n+1)-1$, we estimate
  \begin{equation*}
    \lim_{n\to\infty} \frac{1}{n} \sum_{i=1}^n (\supscr{P}{n})^{k(n)}_{i v_1^{(1)}} \geq
    \lim_{n\to\infty} \frac{1}{n} 2^{k(n)} =
    \lim_{n\to\infty} \frac{n+1}{2n}  = \frac{1}{2},
  \end{equation*}
  where the first inequality follows if one deletes the term relative to
  the self loop in the matrix $\supscr{P}{n}$ and observes that, being all
  nodes but $v_1^{(1)}$ of out-degree $1$, $\sum_{i\neq v_1^{(1)}}
  (\supscr{P}{n})^{k(n)}_{i v_1^{(1)}}$ coincides with the number of paths
  of length $L-1$ from the $L$-th layer to $v_1^{(1)}$.
  
   Figure~\ref{fig:binary-tree} illustrates the behavior of the dynamics in
   the very first steps proving that it is not pre-uniformly wise but that
   the average opinion slowly converges. Because of the computational
   complexity of the simulations (due to the large number of nodes), we are
   unable to clearly illustrate the asymptotic value on this graph for
   large number of layers $L$ or times $k$.\oprocend
    \begin{figure}[H]
    \begin{center}
      \includegraphics[width=.6\linewidth]{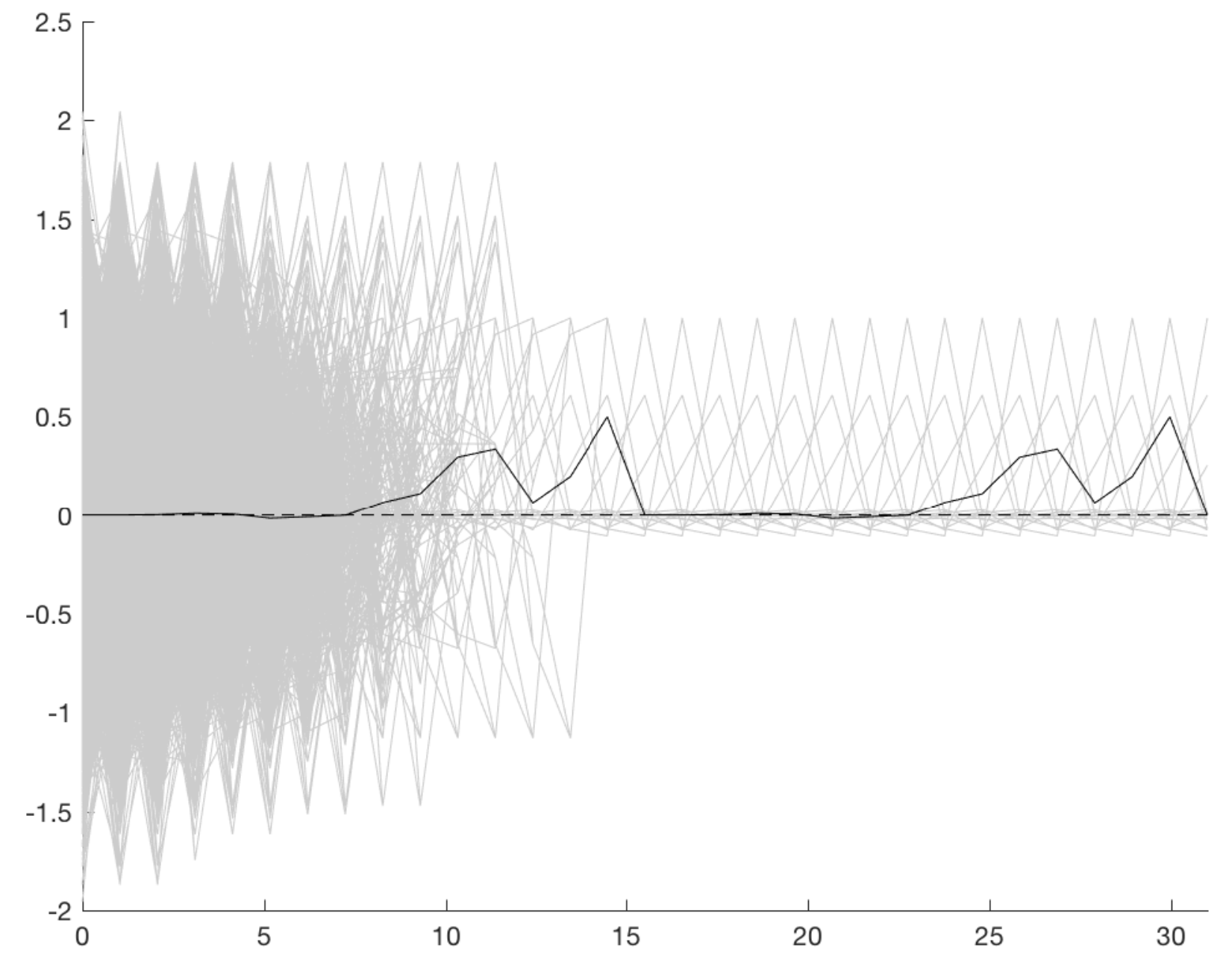}
          \end{center}
          \end{figure}
       \begin{figure}[H]
    \begin{center}
      \includegraphics[width=.6\linewidth]{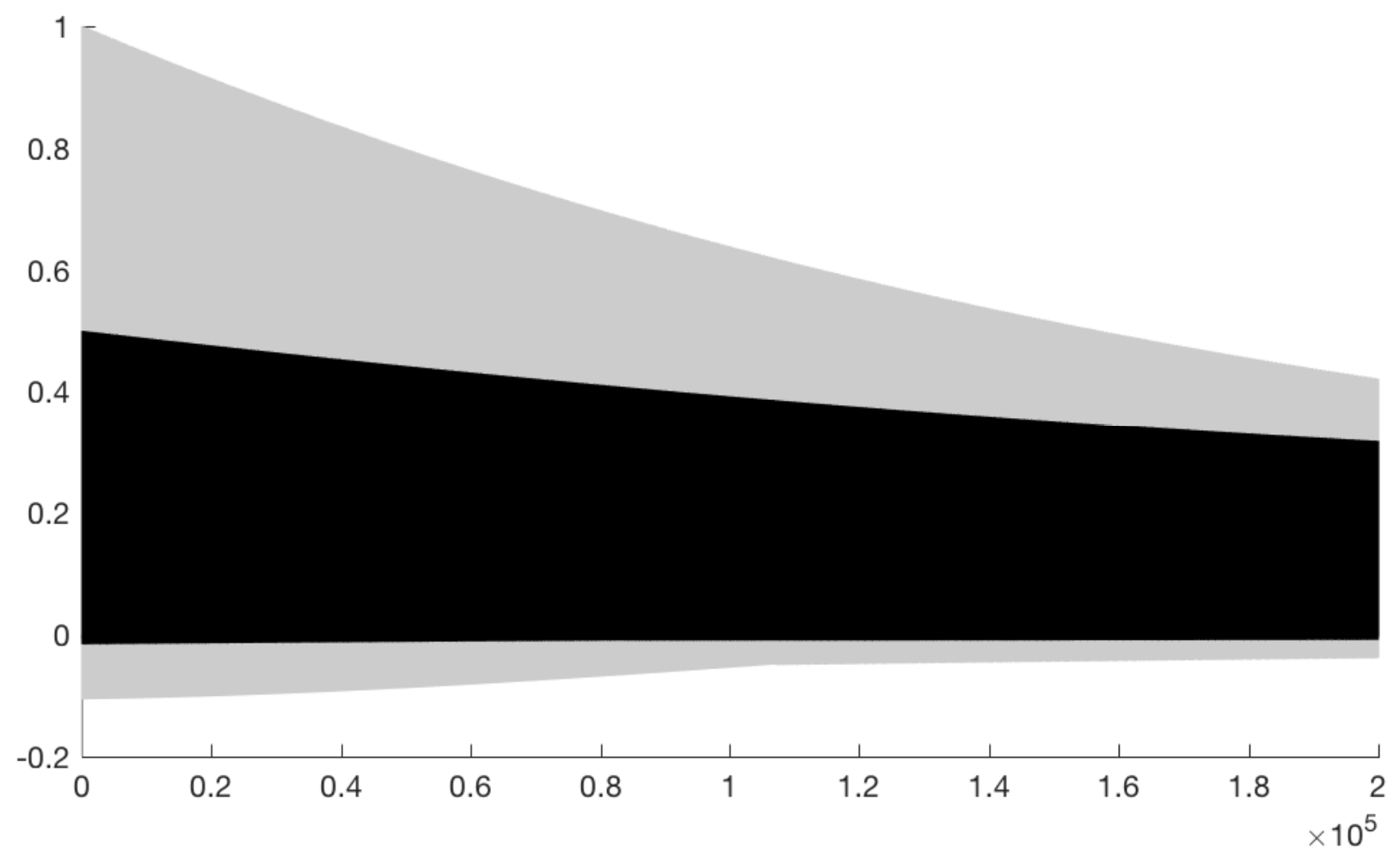}
    \end{center}
    \caption{Reversed binary directed tree with $L=15$ layers
      ($n=2^L-1=32767$ nodes). The time horizon is $T=2L+1=31$ in the first
      plot and $T=2\times 10^5$ in the second. Individual opinions are in
      gray, the average opinion is in black. As proved in
      Example~\ref{ex:binary-tree}, the graph is not pre-uniformly
      wise. Even if hard to illustrate, one can prove that the peaks of the
      average opinion (the black line in the figure) decrease
      asymptotically to zero, confirming the fact that this system is
      wise.}\label{fig:binary-tree}
  \end{figure}
\end{example}


These four examples complete the list of counterexamples in
Lemma~\ref{lemma:basic-implications}. We may make one obvious
additional statement: finite-time wisdom implies one-time wisdom. The
converse is not true, as established by the following counterexample.

\begin{example}[The weighted double-star graph is one-time wise, but not two-time wise]
  \label{ex:double-star}
  We define a sequence of weighted-neighbor matrices (recall
  Definition~\ref{def:equal-neighbor-models}) as follows. For each
  $n$, as illustrated in Figure~\ref{fig:double-star}, the graph is a
  double star with one root node (labeled $j$), $\sqrt{n}$
  intermediate nodes (a representative node is labeled $h$) and $n$
  leafs (a representative node is labeled $i$). For simplicity we
  assume $\sqrt{n}$ is a natural number. The symmetric edge weights
  are selected as follows: $1$ between root and the intermediate nodes
  and $1/\sqrt{n}$ between the intermediate nodes and the leafs.  Note
  that the total number of nodes is $n+\sqrt{n}+1 \asymp n$. We add a
  self-loop with unit weight at the root node $j$ so that each matrix
  in the sequence is primitive.
  \begin{figure}[H]
    \begin{center}
      \includegraphics[width=.6\linewidth]{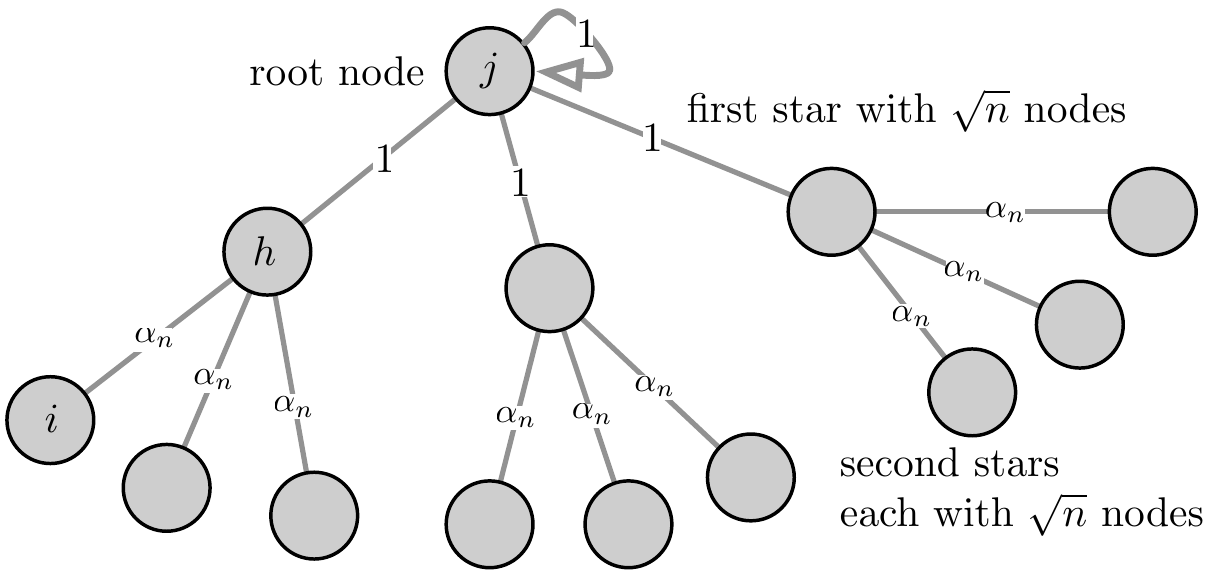}
    \end{center}
    \caption{A weighted double-star graph; we select $\alpha_n=1/\sqrt{n}$.}
    \label{fig:double-star}
  \end{figure}

  With these definitions, we have $W_{hj}=W_{jh}=1$ and, therefore,
  the weighted degree of the root note $j$ is $d_j=\sqrt{n}+1$. We
  also have $W_{ih}=W_{hi}=\alpha_n=1/\sqrt{n}$ and, therefore, the
  weighted degree of each intermediate node $h$ is
  $d_h=1+\sqrt{n}\alpha_n=2$ and the weighted degree of each leaf $i$
  is $d_i=\alpha_n=1/\sqrt{n}$.

  From the equality $P=\diag(W\vect{1}_n)^{-1}W$ defining a
  weighted-neighbor model, we compute $P_{ih}=d_i^{-1} W_{ih}=1$, and
  $P_{jh}= 1/(\sqrt{n}+1)$ so that the column sum of $P$ corresponding
  to an intermediate node $h$ is $\sqrt{n}\times 1 + 1\times
  1/(\sqrt{n}+1) \asymp \sqrt{n}$.  Similarly, we compute $P_{hj}=
  d_h^{-1} W_{hj}=1/2$ so that the column sum of $P$ corresponding to
  the root node $j$ is $\sqrt{n}/2+1/(\sqrt{n}+1)$.  Therefore, the
  sequence is one-time wise.
  
  Finally, for each leaf $i$, we have $(P^2)_{ij}= P_{ih}P_{hj} =
  1/2$. Since there are $n$ leafs, the column sum of~$P^2$
  corresponding to the root node $j$ is at least $n/2 \asymp n$. The
  matrix sequence is therefore not wise at time two.
  
The presence of one-time wisdom and the lack of two-time wisdom is
illustrated in Figure~\ref{fig:double-starSim}. \oprocend
  \begin{center}
    \begin{figure}[H]
      \centering
      \includegraphics[width=.6\linewidth]{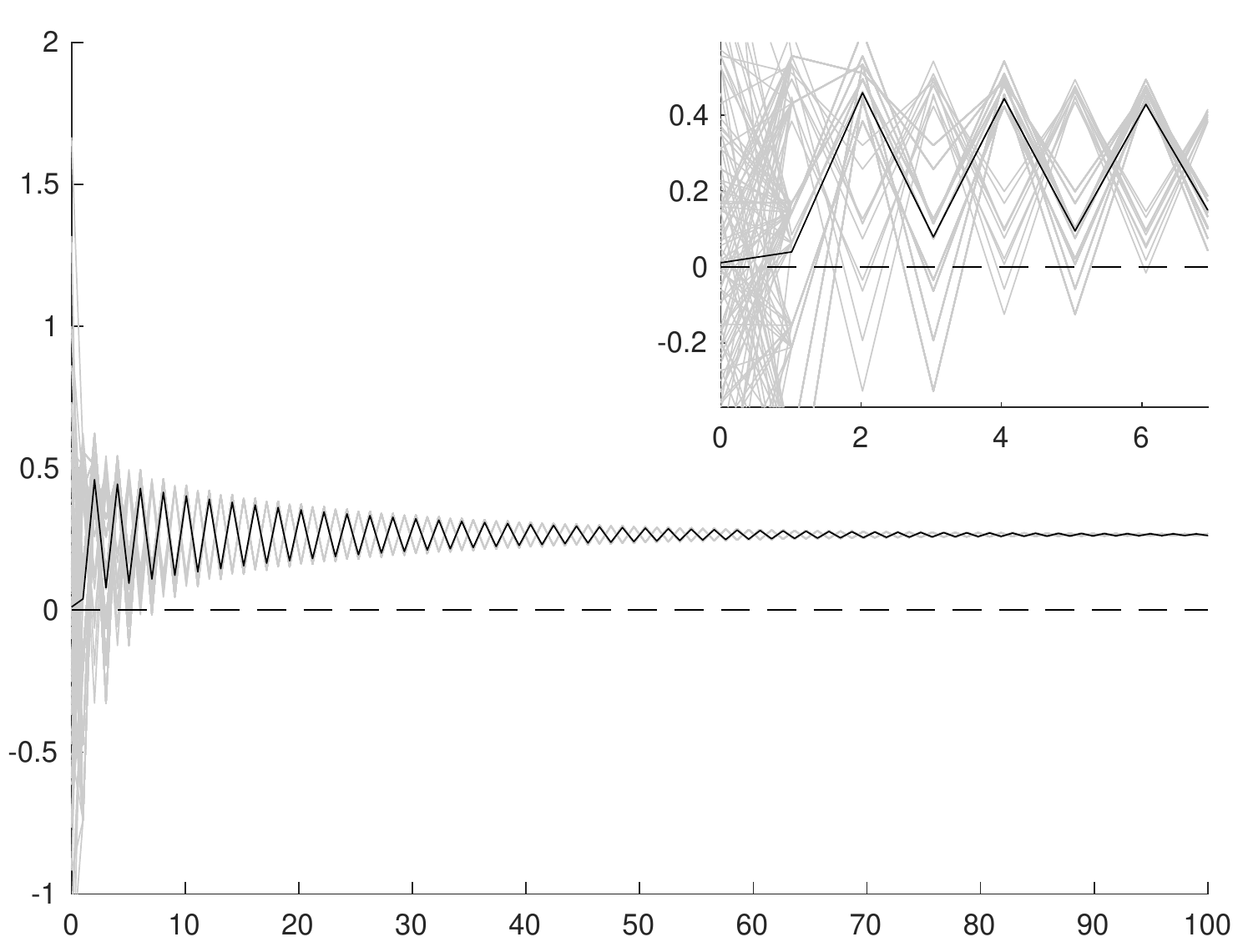}
      \caption{Weighted double star graph with $\sqrt{n}=10$ and therefore
        $n=111$ nodes. The time horizon is $T=100$. Individual opinions are
        in gray, the average opinion is in black.  As proved in
        Example~\ref{ex:double-star}, the two figures illustrate that the
        graph is one time wise (see upper right corner figure) but not two
        time wise (nor wise). Specifically, the theoretical
          analysis shows that the the error at time 1 is of order
          $1/\sqrt{n}$ and at time 2 of order $1$; indeed the black line
          illustrates an error magnitude that is much less than $1/10$ at
          time 1 but not at time 2.  }\label{fig:double-starSim}
    \end{figure}
  \end{center}   
\end{example}


\section{A sufficient condition for finite-time wisdom
in general sequences of stochastic matrices}
\label{sec:prominent-families}

In this section we provide an insightful sufficient condition
guaranteeing finite-time-wisdom for general sequences of matrices. The
main proof idea is to introduce a recursive bound on the total
influence of families of nodes. This bound amounts to a stronger
property than finite-time wisdom and, unlike what occurs for
finite-time wisdom, can be easily checked in terms of the matrix
sequence.

We start with the following useful definition.

\begin{definition}[Prominent families]
  Given a sequence of stochastic matrices of increasing dimensions
  $\{P^{[n]}\in\real^{n\times{n}}\}_{n\in\natural}$, a sequence of
  sets of nodes $\{\supscr{\BB}{n}\subset\until{n}\}_{n\in\natural}$
  is said to be a $P^{[n]}$-\emph{prominent family} if
  \begin{enumerate}
  \item its size is negligible, that is, $|\supscr{\BB}{n}|=o(n)$, but
  \item its \emph{total one-time influence} is order $1$, that is
    \begin{equation*}
      \frac{1}{n} \sum_{i=1}^n \sum_{j\in \supscr{\BB}{n}}
      \supscr{P}{n}_{ij} \asymp 1.
    \end{equation*}
  \end{enumerate}
\end{definition}

While one-time wisdom does not imply finite-time wisdom in general
sequences (see Example~\ref{ex:double-star}), the following key result
shows how the absence of prominent families is inherited by the powers
of $\enne{P}$.

\begin{theorem}[The absence of prominent families is persistent and implies
    finite-time wisdom] \label{theo:prominent} Consider a sequence of
  stochastic matrices of increasing dimensions
  $\{P^{[n]}\in\real^{n\times{n}}\}_{n\in\natural}$.  If there is no
  $\enne{P}$-prominent family, then
  \begin{enumerate}
  \item\label{fact:theo:1} for every $k$, there is no
    $(\enne{P})^k$-prominent family, and
  \item\label{fact:theo:2} $\enne{P}$ is finite-time wise.
  \end{enumerate}
\end{theorem}

In order to prove Theorem~\ref{theo:prominent} we provide some general
notions and then establish the recursive influence bound.  Given an
$n\times n$ stochastic matrix $P$, define the \emph{maximum one-time
  influence of a set of nodes}
$\map{\Phi_P}{\{0,\dots,n\}}{\realnonnegative}$ by
\begin{equation*} 
\begin{cases}
  \Phi_P(0)=0, \\
  \ds \Phi_P(s)=\max\limits_{\BB\subseteq \until{n}, |\BB|=s} \sum\limits_{i=1}^n\sum\limits_{j\in \BB}P_{ij},\quad & s\in\fromto{1}{n}.
\end{cases}
\end{equation*}
Clearly, $\Phi_P(s)$ is non-decreasing in $s$. Notice moreover that
\begin{equation*}
  \begin{split}
    \Phi_P(1)&= \max\limits_{i\in\until{n}}\sum\limits_{i=1}^nP_{ij}=\|P\|_1, \\
    \Phi_P(n)&= \sum_{i=1}^n\sum\limits_{j=1}^nP_{ij}=n.
  \end{split}
\end{equation*}
We can extend the definition of $\Phi_P$ and define
$\map{\Phi_P}{\realnonnegative}{\realnonnegative}$ by
$\Phi_P(x)=\Phi_P(\lfloor x \rfloor).$ Note that $\Phi_P$ remains a
non-decreasing function.

The following two statements are a straightforward consequence of the definition of the function $\Phi_P$.
\begin{remark}\label{remark:prominent}
\begin{enumerate}
  \item\label{fact:prom:1} The absence of $\enne{P}$-prominent
    families is equivalent to the following fact:
    \begin{gather}\label{no_one_prominent}
      \text{for all } \{a_n\in\realnonnegative\}_{n\in\natural} \text{ satisfying } a_n=o(n), \Phi_{\enne{P}}(a_n)=o(n).
    \end{gather}
  \item\label{fact:prom:2} Since
    $n^{-1}\|\enne{P}\|_1=n^{-1}\Phi_{\enne{P}}(1)$, the previous
    statement and Theorem \ref{theo:wisdom} together imply that, if
    there is no $\enne{P}$-prominent family, then $\enne{P}$ is
    one-time wise. \oprocend
\end{enumerate} 
\end{remark}

The following technical result will play a crucial role in our derivations.

\begin{lemma}[Recursive influence bound]\label{lemma:Phi}
  Let $P,Q$ be $n\times n$ stochastic matrices. For every $\delta>0$
  and $s\in\fromto{0}{n}$, it holds
  $$\Phi_{PQ}(s)\leq\Phi_P\left(\delta\Phi_Q(s)\right)+\frac{n}{\delta}.$$
\end{lemma}
\begin{proof}
  Define the shorthand $\mc V=\until{n}$.  Consider any $\BB\subseteq
  \mc V$ and define $$N_\BB^Q(\delta)= \bigsetdef{h\in\mc
    V}{\sum\nolimits_{j\in \BB}Q_{hj}\geq\delta^{-1}}.$$ Notice that
\begin{equation}
\begin{aligned}
  \sum_{i\in\mc V}\sum_{j\in \BB}(PQ)_{ij}
  &=\sum_{i\in\mc V}\sum_{h\in\mc V}\sum_{j\in \BB}P_{ih}Q_{hj}\\
&=\sum_{i\in\mc V}\sum_{h\in N_\BB^Q(\delta)}P_{ih}\sum_{j\in \BB}Q_{hj}+ \sum_{i\in\mc V}\sum_{h\notin N_\BB^Q(\delta)}P_{ih}\sum_{j\in \BB}Q_{hj} \\
&\leq \sum_{i\in\mc V}\sum_{h\in N_\BB^Q(\delta)}P_{ih}+\sum_{i\in\mc V}\sum_{h\notin N_\BB^Q(\delta)}P_{ih}\delta^{-1}\\
&\leq \Phi_P\left(|N_\BB^Q(\delta)|\right)+\frac{n}{\delta},
\end{aligned}
\end{equation}
where last inequality follows from the definition of $\Phi_P$ and the
trivial bound $\sum_{i\in\mc V}\sum_{h\notin
  N_\BB^Q(\delta)}P_{ih}\leq 1$.  Note that
$$\Phi_Q(|\BB|)\geq \sum_{h\in\mc V}\sum_{j\in \BB} Q_{hj}\geq \sum_{h\in N_\BB^Q(\delta)}\sum_{j\in \BB} Q_{hj}\geq |N_\BB^Q(\delta)|\cdot\delta^{-1},$$
which implies
$$|N_\BB^Q(\delta)|\leq \delta\Phi_Q(|\BB|).$$
Therefore,
$$\sum_{i\in\mc V}\sum_{j\in
  \BB}(PQ)_{ij}\leq\Phi_P\left(\delta\Phi_Q(|\BB|)\right)+\frac{n}{\delta}.$$
Now, fix a size $s$ and the proof is complete by computing the maximum
value of the left and right-hand side over all subsets
$\BB\subseteq\mc V$ with $|\BB|=s$.
\end{proof}

We are finally ready to prove the main result in this section.

\begin{proof}[Proof of Theorem~\ref{theo:prominent}]
  We start by proving statement~\ref{fact:theo:1}.  Given
  Remark~\ref{remark:prominent}\ref{fact:prom:1}, it suffices to to
  show that
  \begin{equation} \label{ass:induction:an}
    \text{for all } \{a_n\in\realnonnegative\}_{n\in\natural} \text{ satisfying } a_n=o(n), \enspace \Phi_{(\enne{P})^k}(a_n)=o(n).
  \end{equation}
  We proceed by induction on $k$. Indeed we know from our assumption
  that the statement in equation~\eqref{ass:induction:an} is true for
  $k=1$. We now suppose it is true for $k$ and we prove it for
  $k+1$. We fix a sequence $\{a_n\}$ such that $a_n=o(n)$.
  Lemma~\ref{lemma:Phi} implies
  $$\Phi_{(\enne{P})^{k+1}}(a_n)\leq\Phi_{\enne{P}}\left(\delta\Phi_{(\enne{P})^k}(a_n)\right)+\frac{n}{\delta}.$$
  The induction assumption implies that
  $\Phi_{(\enne{P})^k}(a_n)=o(n)$ and, by
  Remark~\ref{remark:prominent}\ref{fact:prom:1}, we have that
  $$\Phi_{\enne{P}}\left(\delta\Phi_{(\enne{P})^k}(a_n)\right)=o(n).$$
  Therefore
  $$\limsup_{n\to\infty}\frac{\Phi_{(\enne{P})^{k+1})}(a_n)}{n}\leq\frac{1}{\delta}$$
  and, because $\delta$ is arbitrary,
  $$\limsup_{n\to\infty}\frac{\Phi_{(\enne{P})^{k+1})}(a_n)}{n}=0.$$
  This equality completes the proof by induction and proves
  statement~\ref{fact:theo:1}.  Statement~\ref{fact:theo:2} follow
  from statement~\ref{fact:theo:1} considering that
$$\frac{1}{n}\|(\enne{P})^k\|_1=\frac{1}{n}\Phi_{(\enne{P})^k}(1).$$
\end{proof}

\section{A necessary and sufficient condition for pre-uniform
  wisdom in sequences of equal-neighbor matrices}
\label{sec:iff-equal-neighbor}

In this section we focus on sequence of equal-neighbor stochastic
matrices (recall Definition~\ref{def:equal-neighbor-models}).  For
this setting we are able to provide a complete characterization of
one-time, finite-time and pre-uniform wisdom.  We start with a revised
notion of prominence.

\begin{definition}[Prominent individuals]
  Given a sequence of stochastic matrices of increasing dimensions
  $\{P^{[n]}\in\real^{n\times{n}}\}_{n\in\natural}$, a sequence of
  individuals $\{\supscr{k}{n}\in\until{n}\}_{n\in\natural}$ is said
  to be $P^{[n]}$-\emph{prominent} if its \emph{total one-time
    influence} is order $1$, that is
  \begin{equation*}
    \frac{1}{n} \sum_{i=1}^n \supscr{P}{n}_{i \supscr{k}{n}} \asymp 1.
  \end{equation*}
\end{definition}
In other words, a prominent individual is a prominent family composed
of a single individual at each $n$.

\begin{remark} \label{remark:prominent-ind} Some comments are in order. \begin{enumerate}
  \item The absence of prominent families implies the absence of
    prominent individuals; but that the absence of prominent
    individual does not imply the absence of prominent family (e.g.,
    take $\sqrt{n}$ nodes each with $\sqrt{n}$ neighbors with degree
    $1$ in an equal neighbor sequence).
\item The existence of prominent individuals means that there exists a
  sequence of nodes $\{\supscr{k}{n}\in\until{n}\}_{n\in\natural}$
  such that node $\supscr{k}{n}$ possesses order~$n$ neighbors with
  order~$1$ degree.\oprocend
  \end{enumerate}
\end{remark}

We next show how the absence of prominent individuals and the notion
of one-time wisdom play a key role in equal-neighbor sequences.

\begin{theorem}[The absence of prominent individuals is necessary
    and sufficient for pre-uniform wisdom in equal-neighbor
    sequences] \label{thm:prominent-equal} Consider a sequence of
  equal-neighbor matrices of increasing dimensions
  $\{P^{[n]}\}_{n\in\natural}$.  The following statements are
  equivalent:
  \begin{enumerate}
  \item\label{fact:pe:prominent-individual} there is no
    $P^{[n]}$-prominent individual,
  \item\label{fact:pe:one-time-wise} the sequence is one-time wise, and
  \item\label{fact:pe:uniformly-wise} the sequence is pre-uniformly
    wise.
  \end{enumerate}
\end{theorem}

In other words, one-time wisdom implies finite-time wisdom and
pre-uniform wisdom for the setting of equal-neighbor sequences as well
as wisdom for the setting of primitive equal-neighbor
sequences. Recall that for more general sequences, e.g., the setting
of weighted-neighbor models, one-time wisdom does not imply two-time
wisdom (see Example~\ref{ex:double-star}).

Based on Theorem~\ref{thm:prominent-equal} (and specifically on the
fact that one-time wisdom implies wisdom) and
Example~\ref{ex:complete+star} (showing an example of a wise but not
one-time wise sequence), we classify the equal-neighbor sequences as
shown in Figure~\ref{fig:equal-neighbor-sets}.

\begin{figure*}[hbt]
  \begin{center}
    \includegraphics[width=\linewidth]{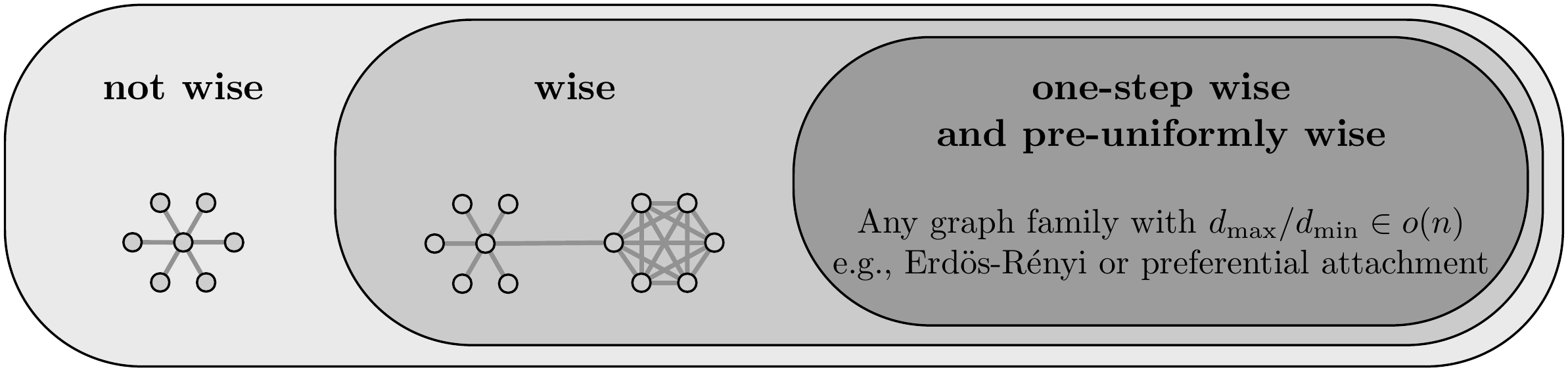}
  \end{center}
  \caption{Logical relations among sets of equal-neighbor sequences
    with varying degree of wisdom. For each set we provide an example
    of a equal-neighbor sequence.}
  \label{fig:equal-neighbor-sets}
\end{figure*}

In order to prove Theorem~\ref{thm:prominent-equal} we establish some
useful bounds.

\begin{lemma}[One-norm bounds]
  \label{lemma:uniform-onenorm-bound}
  Any equal-neighbor matrix $P\in\real^{n\times{n}}$ satisfies, for
  all $k\in\natural$,
  \begin{align}
    \bigonenorm{\tfrac{1}{n} P^k}  &\leq 2 \bigonenorm{\tfrac{1}{n} P}^{1/2}.  \label{bound:main} 
  \end{align}
\end{lemma}

\begin{remark}
 The square root is unavoidable, i.e., it is not possible to obtain a
 bound similar to~\eqref{bound:main} for equal-neighbor matrices that
 does not involve the square root of the one-norm of $P$. For example,
 the double-star with $\sqrt{n}$ aperture is an example where
 $\onenorm{\tfrac{1}{n} P} = 1/n$ and $\onenorm{\tfrac{1}{n} P^2} =
 1/\sqrt{n}$. \oprocend
\end{remark}

\begin{proof}[Proof of Lemma~\ref{lemma:uniform-onenorm-bound}]
 Let $P=\diag(W\vect{1}_n)^{-1}W$ for a binary and symmetric $W$, let
 $G$ denote the undirected graph defined by $W$, pick a node
 $j\in\until{n}$ and a constant $\delta>0$. Define
  \begin{equation}
    \label{def:Njdelta}
    N_j(\delta) = \Bigsetdef{h\in\until{n}}{P_{hj}\geq \frac{1}{\delta}} = \setdef{h\in N_j}{d_h\leq\delta}.
  \end{equation}
  We now compute
  \begin{align*}
    P^k_{ij} &= \sum_{h=1}^n P^{k-1}_{ih}P_{hj} =
    \sum_{h\in N_j(\delta)} P^{k-1}_{ih}P_{hj} +
    \sum_{h\notin N_j(\delta)} P^{k-1}_{ih}P_{hj} \\
    &\leq  \sum_{h\in N_j(\delta)} P^{k-1}_{ih}P_{hj} + \sum_{h\notin N_j(\delta)} P^{k-1}_{ih}\frac{1}{\delta}
    \\
    &=\sum_{h\in N_j(\delta)} P^{k-1}_{ih}P_{hj} + \frac{1}{\delta}.
  \end{align*}
  We now introduce the shorthand $\gamma_j^{(k)}=\sum_{i=1}^n
  P^k_{ij}$ for the $j$-th column sum of $P^k$ and obtain
  \begin{align}
    \label{eq:tmp-gamma}
    \gamma_j^{(k)} = \sum_{i=1}^n P^k_{ij} &\leq \sum_{i=1}^n \sum_{h\in N_j(\delta)} P^{k-1}_{ih}P_{hj} + \frac{n}{\delta}.
  \end{align}
  We change the order of summation and compute, for any $h\in
  N_j(\delta)$,
  \begin{align*}
    \sum_{i=1}^n P^{k-1}_{ih} 
    &=
    \sum_{i=1}^n  \sum_{\ell_1=1}^n \cdots \sum_{\ell_{k-2}=1}^n 
    (d_i^{-1}W_{i\ell_1})   (d_{\ell_1}^{-1}W_{\ell_1\ell_2})
    \cdots
    (d_{\ell_{k-2}}^{-1}W_{\ell_{k-2}h}) \\
    &=    \sum_{i=1}^n \sum_{\ell_1=1}^n  \cdots  \sum_{\ell_{k-2}=1}^n d_i^{-1} 
    (W_{\ell_1i} d_{\ell_1}^{-1}) (W_{\ell_2\ell_1} d_{\ell_2}^{-1})
    \cdots
    (W_{\ell_{k-2}\ell_{k-3}} d_{\ell_{k-2}}^{-1}) W_{h\ell_{k-2}},
  \end{align*}
  where we use the symmetry of $W$, and reorganize the products inside
  the summation. We now upper bound $d_i^{-1}$ with $1$, note that
  $\sum_{i=1}^n W_{\ell i} d_{\ell}^{-1} = 1$, and change the order of
  summation so that
  \begin{align*}
    \sum_{i=1}^n P^{k-1}_{ih} 
    & \leq  \sum_{\ell_{k-2}=1}^n \cdots \sum_{\ell_1=1}^n
    \sum_{i=1}^n  
    (W_{\ell_1i} d_{\ell_1}^{-1}) (W_{\ell_2\ell_1} d_{\ell_2}^{-1}) \cdots
    (W_{\ell_{k-2}\ell_{k-3}} d_{\ell_{k-2}}^{-1}) W_{h\ell_{k-2}}
    \\
    &= \sum_{\ell_{k-2}=1}^n W_{h\ell_{k-2}}  = d_h.
  \end{align*}
  We plug this inequality into~\eqref{eq:tmp-gamma} and adopt some additional bounds to obtain
  \begin{align*} 
    \gamma_j^{(k)} &\leq \sum_{h\in N_j(\delta)} d_h P_{hj}
    + \frac{n}{\delta} \leq \Big(\max_{h\in N_j(\delta)} d_h\Big) \gamma_j^{(1)}+ \frac{n}{\delta}.
  \end{align*}
  Because $W$ is binary, we know that $h\in N_j(\delta)$ implies
  $d_h\leq \delta$ (see also definition~\eqref{def:Njdelta}). The last
  inequality becomes:
  \begin{align*} 
    \frac{1}{n}\gamma_j^{(k)} &\leq \delta \frac{1}{n} \gamma_j^{(1)}+ \frac{1}{\delta}.
  \end{align*}
  By selecting $\delta = \Big(\frac{1}{n}\gamma_j^{(1)}\Big)^{1/2}$, we obtain that
  each column average of the $P^k$ satisfies
  \begin{align*} 
    \frac{1}{n}\gamma_j^{(k)} &\leq 2 \Big(\frac{1}{n}\gamma_j^{(1)}\Big)^{1/2}.
  \end{align*}
  This inequality immediately implies inequality~\eqref{bound:main} in the theorem statement.
\end{proof}  

We are now ready to prove the main result of this section.

\begin{proof}[Proof of Theorem~\ref{thm:prominent-equal}]
    We start by showing the equivalence \ref{fact:pe:one-time-wise} $\iff$
    \ref{fact:pe:uniformly-wise}.  The fact that one-time wisdom
    (statement~\ref{fact:pe:one-time-wise}) implies pre-uniform wisdom
    (statement~\ref{fact:pe:uniformly-wise}) for equal-neighbor sequences
    is a direct consequence of inequality~\eqref{bound:main} in
    Lemma~\ref{lemma:uniform-onenorm-bound}.  The converse implication
    \ref{fact:pe:uniformly-wise} $\implies$ \ref{fact:pe:one-time-wise} is
    trivially true and stated in Lemma~\ref{lemma:basic-implications}.

  Next, we show the equivalence between the absence of prominent
  individuals (statement~\ref{fact:pe:prominent-individual}) and
  one-time wisdom (statement~\ref{fact:pe:one-time-wise}). To do so,
  we prove the equivalence between the existence of prominent
  individuals and the lack of one-time wisdom.
  First, assume there exists a prominent individual, that is, as
  mentioned in Remark~\ref{remark:prominent-ind}, a sequence of nodes
  $\{\supscr{k}{n}\in\until{n}\}_{n\in\natural}$ such that node
  $\supscr{k}{n}$ possesses order~$n$ neighbors in $\supscr{G}{n}$
  with order~$1$ degree.  In other words there exists a constant
  $\delta>0$ such that, recalling the definition in
  equation~\eqref{def:Njdelta}, we have $|N_{\supscr{k}{n}}(\delta)|$
  is of order $n$.  Then, from the equality~\eqref{def:Pij=diinv-Wij},
  \begin{align*}
    \frac{1}{n} \onenorm{\supscr{P}{n}}
    &= \frac{1}{n}\max_{j\in\until{n}} \sum_{i=1}^n    \supscr{P}{n}_{ij} \geq
    \frac{1}{n} \sum_{i=1}^n \supscr{P}{n}_{i\supscr{k}{n}}\\
    & = \frac{1}{n}  \sum_{i \in N_{\supscr{k}{n}}} d_i^{-1} \supscr{W}{n}_{i \supscr{k}{n}}
    \geq \frac{1}{n \delta}  \sum_{i \in   N_{\supscr{k}{n}}(\delta) } \supscr{W}{n}_{i \supscr{k}{n}}
    \\
    &= \frac{1}{n \delta}  |N_{\supscr{k}{n}}(\delta)| .
  \end{align*}
  This inequality shows that $\lim_{n\to\infty} \onenorm{\frac{1}{n}
    P^{[n]}} $ cannot vanish and, therefore, the sequence of
  equal-neighbor matrices fails to be one-time wise.

  Second, assume the sequence of equal-neighbor matrices fails to be
  one-time wise. Then there exist a constant $\alpha>0$ and a time $N$
  such that, for all $n>N$,
  \begin{equation*}
    \frac{1}{n}\max_{j\in\until{n}} \sum_{i=1}^n \supscr{P}{n}_{ij} \geq \alpha.
  \end{equation*}
  In other words, there must exist a sequence of indices
  $\{\supscr{k}{n}\in\until{n}\}_{n\in\natural}$ such that $\sum_{i
    \in N_{\supscr{k}{n}}} d_i^{-1}\supscr{W}{n}_{i\supscr{k}{n}}
  >\alpha n$.  Because each degree is lower bounded by $1$, this
  inequality implies the existence of a prominent individual.
\end{proof}

\begin{remark}[Wisdom estimates for large by finite populations]
  Lemma~\ref{lemma:uniform-onenorm-bound} leads to an explicit estimation
  of the rate at which finite time wisdom is achieved, namely the rate of
  convergence of $\ave(\supscr{x}{n}(k)))$ to $\mu$. Indeed, from the proof
  of Theorem~\ref{theo:wisdom} (equality~\eqref{unbiased} and the following
  steps) and inequality~\eqref{bound:main}, we can estimate
  \begin{equation*}
    \Var[\ave(\supscr{x}{n}(k))]\leq \sigma^2
    \bigonenorm{\tfrac{1}{n}(P^{[n]})^k}\leq 2\sigma^2\bigonenorm{\tfrac{1}{n} P^{[n]}}^{1/2}.
  \end{equation*}
  Similar considerations can also be applied to the concept of uniform
  wisdom, working directly this time with the deviation probability. From
  the proof of Theorem~\ref{theo:sufficient-uniform} and applying again
  inequality~\eqref{bound:main}, we obtain
  \begin{equation*}
    \Prob\left[\sup\limits_{k\in\natural}
      |\ave(\supscr{x}{n}(k)))-\mu|\geq\delta\right]\leq
    \frac{18e\sigma^2}{\delta^2}\bigonenorm{\tfrac{1}{n} P^{[n]}}^{1/2}\taumix(\enne{P}).
  \end{equation*}
  These bounds allow us to quantify how close to wisdom (finite or uniform)
  is a large but finite group of individuals whose opinions evolve
  according to an equal neighbor French-DeGroot model. They are
  particularly effective in those cases when the quantities
  $\bigonenorm{P^{[n]}}$ and $\taumix(\enne{P})$ can be estimated, as it is
  the case for the \Erdos-\Renyi\ graphs and the preferential attachment
  model studied in Subsection~\ref{examples}.
\end{remark}





\newcommand{\dmax}{\subscr{d}{max}}
\newcommand{\dmin}{\subscr{d}{min}}

\newcommand{\dmaxn}{\enne{\subscr{d}{max}}}
\newcommand{\dminn}{\enne{\subscr{d}{min}}}

\subsection{The equal-neighbor model over prototypical random graphs}\label{examples}
For equal-neighbor models, a sufficient condition to guarantee
one-time and pre-uniformly wisdom is:
\begin{equation}
  \label{cond:dmaxmin}
  \frac{\dmax(n)}{\dmin(n)} =  o(n),
\end{equation}
where $\dmax(n)$ and $\dmin(n)$ denote, respectively, the maximum and
minimum degree of the graph as a function of the network size $n$.
Indeed, let $\enne{k}$ be a vertex of the graph and compute:
\begin{align*}
  \frac{1}{n}\sum_{i=1}^n\enne{P}_{i\enne{k}}\leq
  \frac{1}{n}\sum_{i=1}^n\frac{1}{\dmin(n)} \leq
  \frac{1}{n}\frac{\dmax(n)}{\dmin(n)},
\end{align*}
as $n\to\infty$. If condition~\eqref{cond:dmaxmin} holds, then
$\frac{1}{n}\sum_{i=1}^n\enne{P}_{i\enne{k}}\to0$ as $n\to\infty$ so
that there exist no prominent individuals.

In this section we study the \Erdos-\Renyi\ and the \Barabasi-Albert
preferential attachment models of random graph and we show that, using
condition~\eqref{cond:dmaxmin}, they are both one-time wise with high
probability.  In contexts where we have a sequence of probability
spaces labeled by parameter $n$ (in our case the number of nodes in
the graph), the locution with high probability (w.h.p.) means with
probability converging to $1$ as $n\to +\infty$. In the case of a
limit property, as the case of one-time wise, to assert that it holds
w.h.p.\ means that, for every $\epsilon >0$,
$\Prob[\onenorm{\frac{1}{n} P^{[n]}} <\epsilon]$ converges to $1$ for
$n\to +\infty$.

\begin{example}[The equal-neighbor model over \Erdos-\Renyi\ graphs]
An \Erdos-\Renyi\ graph $G(n,p)$ is a graph with $n$ vertices and with
each possible edge having, independently, probability $p$ of existing
\cite{PE-AR:60}. We focus on the case when $p={c \log (n)}/{n}$
with $c>1$. In this regime $G(n,p)$ is known to be connected and
aperiodic w.h.p.. Moreover, \cite[Lemma 6.5.2]{RD:06} implies that
there exists a constant $b>0$ such that $b\log n\leq d_i \leq 4c\log
n$ for every node $i$ w.h.p..  This bound immediately implies that
condition~\eqref{cond:dmaxmin} holds w.h.p.\ and thus the
\Erdos-\Renyi\ model is one-time wise (and pre-uniformly wise and
wise) w.h.p.. Using the fact that~\cite{RD:06} $\taumix(\enne{P})=
O(\log (n))$ w.h.p., we now prove that this model is also uniformly
wise. Indeed, w.h.p.
 \begin{equation*}
   \begin{aligned}
     \bigonenorm{\tfrac{1}{n} (\enne{P})^k} \taumix(\enne{P})
     &\leq 2\bigonenorm{\tfrac{1}{n} \enne{P}}^{1/2} \taumix(\enne{P})\\
     &= O\left(\frac{1}{n^{\frac{1}{2}}}\log n \right),
     \quad\text{as $n\to\infty$,}
    \end{aligned}
  \end{equation*}
where the first inequality follow from
equation~\eqref{bound:main}. \oprocend
\end{example}


We now present the preferential attachment model. Also in this case, it is
possible to prove that the equal-neighbor models is one time wise and uniformly
wise w.h.p..

\begin{example}[The equal-neighbor model over preferential attachment graphs]
  The \Barabasi-Albert preferential attachment model is a random graph
  generation model described as follows: vertices are added
  sequentially to the graph, new vertices are connected to a fixed
  number of earlier vertices, that are selected with probabilities
  proportional to their degrees.
  Specifically, assume a fixed number $m_0$ of initial vertices is
  given and, at every step, a new vertex is added and $m$ ($m\leq
  m_0$) new edges are added, whereby the new vertex is connected to a
  prior node $i$ with a probability proportional to the degree $d_i$
  of $i$, that is $d_i/\sum_jd_j$.  After $t$ time steps, the model
  leads to a random graph with $t+m_0$ vertices and $mt$ edges.  It is
  known~\cite{ALB-RA:99,BB-OR-JS-GT:01} that the degree distribution
  follows a power-law, that the minimum degree is $\dmin(n)=m$ (by
  construction), and that the maximum degree is of order
  $\dmax(n)\in\Theta(\sqrt{n})$ w.h.p..
  
  The equal-neighbor \Barabasi-Albert model is one-time wise (and,
  therefore, also pre-uniformly wise and wise) w.h.p.. This follows
  again by checking condition~\eqref{cond:dmaxmin}:
   \begin{equation*}
    \frac{1}{n}\frac{\dmax(n)}{\dmin(n)}=O\left(
    \frac{1}{n}\frac{\sqrt n}{m}\right)=O\left( \frac{1}{m \sqrt n
    }\right), \quad \text{as $n\to\infty$.}
  \end{equation*}
 
  Moreover, the equal-neighbor \Barabasi-Albert model is uniformly
  wise. Indeed, recall from~\cite{DA-GC-FF-AO:10} that the mixing time
  of the \Barabasi-Albert model is
  w.h.p.\ $\taumix(\enne{P})=O(\log(n))$, so that w.h.p.
  \begin{equation*}
  \begin{aligned}
    \bigonenorm{\tfrac{1}{n} (\enne{P})^k} \taumix(\enne{P})
    &\leq 2\bigonenorm{\tfrac{1}{n} \enne{P}}^{1/2} \taumix(\enne{P})\\
    &= O\left(\frac{1}{n^{1/4}}\log n\right), \quad \text{as $n\to\infty$}, 
    \end{aligned}
  \end{equation*}
  where the first inequality follows from
  inequality~\eqref{bound:main}.
  
  Finally, we consider a super-linear preferential attachment model.
  Notice that the \Barabasi-Albert model is a linear preferential
  attachment in the sense that the probability of choosing a node in
  the network is linear in the degree of the nodes. If we consider a
  \emph{super-linear model} with a probability of the form $\sim x^p$
  with $p>1$, then it is known~\cite{RO-JS:05} that there exists,
  w.h.p., a node with degree of order $n$, while all other nodes have
  finite degrees. It follows that a sequence of prominent individuals
  exists in large populations and that, by
  Theorem~\ref{thm:prominent-equal}, the super-linear preferential
  attachment model is neither wise nor finite-time wise. \oprocend
\end{example}

\section{Conclusions}
\label{sec:conclusions}

This paper furthers the study of learning phenomena and influence
systems in large populations. Our results provide an alternative and,
arguably, a bit more realistic characterization of wise populations in
terms of the absence of prominently influential individuals and
groups. Future work includes extending these concepts to influence
systems with time-varying and concept-dependent interpersonal weights
and to other opinion dynamic models.

\section{Acknowledgments}
The first author thanks Dr.\ Noah E.\ Friedkin for an early inspiring
discussion about \naive learning. This material is based upon work
supported by, or in part by, the U.S.\ Army Research Laboratory and the
U.S.\ Army Research Office under grant number W911NF-15-1-0577.

\bigskip\bigskip\bigskip
\bibliographystyle{plainurl}
\bibliography{main_arxiv_v2}

\begin{thebibliography}{10}

\bibitem{DA-GC-FF-AO:10}
D.~Acemoglu, G.~Como, F.~Fagnani, and A.~Ozdaglar.
\newblock Opinion fluctuations and disagreement in social networks.
\newblock {\em Mathematics of Operation Research}, 38(1):1--27, 2013.
\newblock \href {http://dx.doi.org/10.1287/moor.1120.0570}
  {\path{doi:10.1287/moor.1120.0570}}.

\bibitem{DA-MAD-IL-AO:11}
D.~Acemoglu, M.~A. Dahleh, I.~Lobel, and A.~Ozdaglar.
\newblock Bayesian learning in social networks.
\newblock {\em Review of Economic Studies}, 78(4):1201--1236, 2011.
\newblock \href {http://dx.doi.org/10.1093/restud/rdr004}
  {\path{doi:10.1093/restud/rdr004}}.

\bibitem{ALB-RA:99}
A.-L. Barab{\'a}si and R.~Albert.
\newblock Emergence of scaling in random networks.
\newblock {\em Science}, 286(5439):509--512, 1999.
\newblock \href {http://dx.doi.org/10.1126/science.286.5439.509}
  {\path{doi:10.1126/science.286.5439.509}}.

\bibitem{JB-DB-DC:17}
J.~Becker, D.~Brackbill, and D.~Centola.
\newblock Network dynamics of social influence in the wisdom of crowds.
\newblock {\em Proceedings of the National Academy of Sciences},
  114(26):E5070--E5076, 2017.
\newblock \href {http://dx.doi.org/10.1073/pnas.1615978114}
  {\path{doi:10.1073/pnas.1615978114}}.

\bibitem{BB-OR-JS-GT:01}
B.~Bollob{\'a}s, O.~Riordan, J.~Spencer, and G.~Tusn{\'a}dy.
\newblock The degree sequence of a scale-free random graph process.
\newblock {\em Random Structures \& Algorithms}, 18(3):279--290, 2001.
\newblock \href {http://dx.doi.org/10.1002/rsa.1009}
  {\path{doi:10.1002/rsa.1009}}.

\bibitem{PB:72c}
P.~Bonacich.
\newblock Factoring and weighting approaches to status scores and clique
  identification.
\newblock {\em Journal of Mathematical Sociology}, 2(1):113--120, 1972.
\newblock \href {http://dx.doi.org/10.1080/0022250X.1972.9989806}
  {\path{doi:10.1080/0022250X.1972.9989806}}.

\bibitem{MHDG:74}
M.~H. DeGroot.
\newblock Reaching a consensus.
\newblock {\em Journal of the American Statistical Association},
  69(345):118--121, 1974.
\newblock \href {http://dx.doi.org/10.1080/01621459.1974.10480137}
  {\path{doi:10.1080/01621459.1974.10480137}}.

\bibitem{PMD-DV-JZ:03}
P.~M. DeMarzo, D.~Vayanos, and J.~Zwiebel.
\newblock Persuasion bias, social influence, and unidimensional opinions.
\newblock {\em Quarterly Journal of Economics}, 118(3):909--968, 2003.
\newblock \href {http://dx.doi.org/10.1162 /00335530360698469}
  {\path{doi:10.1162 /00335530360698469}}.

\bibitem{RD:06}
R.~Durrett.
\newblock {\em Random Graph Dynamics}.
\newblock Cambridge University Press, 2006.
\newblock \href {http://dx.doi.org/10.1017/CBO9780511546594}
  {\path{doi:10.1017/CBO9780511546594}}.

\bibitem{PE-AR:60}
P.~Erd\"os and A.~R{\'e}nyi.
\newblock On the evolution of random graphs.
\newblock {\em Publication of the Mathematical Institute of the Hungarian
  Academy of Science}, 5(1):17--60, 1960.

\bibitem{JRPF:56}
J.~R.~P. French.
\newblock A formal theory of social power.
\newblock {\em Psychological Review}, 63(3):181--194, 1956.
\newblock \href {http://dx.doi.org/10.1037/h0046123}
  {\path{doi:10.1037/h0046123}}.

\bibitem{NEF:91}
N.~E. Friedkin.
\newblock Theoretical foundations for centrality measures.
\newblock {\em American Journal of Sociology}, 96(6):1478--1504, 1991.
\newblock \href {http://dx.doi.org/10.1086/229694} {\path{doi:10.1086/229694}}.

\bibitem{NEF-ECJ:11}
N.~E. Friedkin and E.~C. Johnsen.
\newblock {\em Social Influence Network Theory: {A} Sociological Examination of
  Small Group Dynamics}.
\newblock Cambridge University Press, 2011.

\bibitem{FG:1907}
F.~Galton.
\newblock Vox populi.
\newblock {\em Nature}, 75:450--451, 1907.
\newblock \href {http://dx.doi.org/10.1038/075450a0}
  {\path{doi:10.1038/075450a0}}.

\bibitem{BG-MOJ:10}
B.~Golub and M.~O. Jackson.
\newblock Na\"ive learning in social networks and the wisdom of crowds.
\newblock {\em American Economic Journal: Microeconomics}, 2(1):112--149, 2010.
\newblock \href {http://dx.doi.org/10.1257/mic.2.1.112}
  {\path{doi:10.1257/mic.2.1.112}}.

\bibitem{MOJ:10}
M.~O. Jackson.
\newblock {\em Social and Economic Networks}.
\newblock Princeton University Press, 2010.

\bibitem{AJ-AS-ATS:10}
A.~Jadbabaie, A.~Sandroni, and A.~Tahbaz-Salehi.
\newblock Non-{B}ayesian social learning.
\newblock {\em Games and Economic Behavior}, 76(1):210--225, 2012.
\newblock \href {http://dx.doi.org/10.1016/j.geb.2012.06.001}
  {\path{doi:10.1016/j.geb.2012.06.001}}.

\bibitem{DAL-YP-ELW:09}
D.~A. Levin, Y.~Peres, and E.~L. Wilmer.
\newblock {\em Markov Chains and Mixing Times}.
\newblock American Mathematical Society, 2009.

\bibitem{JL-HR-FS-DH:11}
J.~Lorenz, H.~Rauhut, F.~Schweitzer, and D.~Helbing.
\newblock How social influence can undermine the wisdom of crowd effect.
\newblock {\em Proceedings of the National Academy of Sciences},
  108(22):9020--9025, 2011.
\newblock \href {http://dx.doi.org/10.1073/pnas.1008636108}
  {\path{doi:10.1073/pnas.1008636108}}.

\bibitem{RO-JS:05}
R.~Oliveira and J.~Spencer.
\newblock Connectivity transitions in networks with super-linear preferential
  attachment.
\newblock {\em Internet Mathematics}, 2(2):121--163, 2005.
\newblock \href {http://dx.doi.org/10.1080/15427951.2005.10129101}
  {\path{doi:10.1080/15427951.2005.10129101}}.

\bibitem{AVP-RT:17}
A.~V. Proskurnikov and R.~Tempo.
\newblock A tutorial on modeling and analysis of dynamic social networks. {Part
  I}.
\newblock {\em Annual Reviews in Control}, 43:65--79, 2017.
\newblock \href {http://dx.doi.org/10.1016/j.arcontrol.2017.03.002}
  {\path{doi:10.1016/j.arcontrol.2017.03.002}}.

\bibitem{WEP:66}
W.~E. Pruitt.
\newblock Summability of independent random variables.
\newblock {\em Journal of Mathematics and Mechanics}, 15:769--776, 1966.
\newblock \href {http://dx.doi.org/10.1512/iumj.1966.15.15052}
  {\path{doi:10.1512/iumj.1966.15.15052}}.

\bibitem{JSU:04}
J.~Surowiecki.
\newblock {\em The Wisdom Of Crowds}.
\newblock Anchor, 2004.

\end{thebibliography}

\end{document}